\documentclass[11pt]{article}
\usepackage[english]{babel}
\usepackage{dsfont}
\usepackage{graphicx}
\usepackage{amsmath}
\usepackage{amssymb}
\usepackage{amsthm}
\usepackage{natbib}
\usepackage{epstopdf}
\usepackage{geometry}
\usepackage{booktabs}

\allowdisplaybreaks[2]

\usepackage[dvipsnames]{xcolor}
\usepackage{hyperref}
\hypersetup{pdfauthor={Denis Chetverikov and Daniel Wilhelm},colorlinks=true,citecolor=Blue,filecolor=Blue,linkcolor=Blue,urlcolor=Blue}

\usepackage[colorinlistoftodos,prependcaption,textsize=tiny]{todonotes}

 \usepackage[font=footnotesize,labelfont=bf]{caption}

\theoremstyle{plain}
\newtheorem{theorem}{Theorem}
\newtheorem{lemma}{Lemma}

\newtheorem{assumption}{Assumption}

\theoremstyle{definition}
\newtheorem{remark}{Remark}

\newcommand*{\QEDA}{\hfill\ensuremath{\square}}

\newcommand{\ind}{\mathds{1}}
\DeclareMathOperator*{\plim}{plim}

\renewcommand{\hat}{\widehat}
\renewcommand{\tilde}{\widetilde}

\usepackage{xr}

\begin{document}
\title{Inference for Rank-Rank Regressions\thanks{We thank Jin Hahn and Andres Santos for helpful discussions and Albert Belkov, Fariba Dorpoush, Owen Groves, Philipp Kropp, Pawel Morgen, and Nicholas Vieira for excellent research assistance.}}

\author{Denis Chetverikov\thanks{Department of Economics, University of California at Los Angeles, 315 Portola Plaza, Bunche Hall, Los Angeles, CA 90024, USA; E-Mail address: \href{mailto:chetverikov@econ.ucla.edu}{\texttt{chetverikov@econ.ucla.edu}}.} \and Daniel Wilhelm\thanks{Departments of Statistics and Economics, Ludwig-Maximilians-Universität München, Akademiestr. 1, 80799 Munich, Germany; E-Mail address: \href{mailto:d.wilhelm@lmu.de}{\texttt{d.wilhelm@lmu.de}}. The author gratefully acknowledges financial support from the European Research Council (Starting Grant No. 852332)}\vspace{1cm}}
\maketitle
\thispagestyle{empty}
\begin{abstract}
    The slope coefficient in a rank-rank regression is a popular measure of intergenerational mobility. In this article, we first show that commonly used inference methods for this slope parameter are invalid. Second, when the underlying distribution is not continuous, the OLS estimator and its distribution may be sensitive to how ties in the ranks are handled. Motivated by these findings we derive the asymptotic distribution of the OLS estimator and provide valid inference methods for a general class of rank-rank regression specifications. 
\end{abstract}

\newpage

\section{Introduction}
Regressions involving ranks are widely used in empirical work in economics. A notable example is a rank-rank regression for measuring the persistence in socioeconomic status across generations. \cite{Dahl:2008tg} and \cite{Chetty:2014tr} have been influential in promoting this approach and a vast and fast-growing empirical literature is using such and related regressions for the study of intergenerational mobility in different socioeconomic outcomes, countries, regions, and time periods.\footnote{Some examples are \cite{Olivetti:2015tt}, \cite{Black:2019oi}, \cite{Abramitzky:2021ii}, \cite{Fagereng:2021wb}, \cite{Nybom:2024aa}, and \cite{Ward:2023uu}. Recent surveys by \cite{Deutscher:2023oo} and \cite{Mogstad:2023uu} provide numerous further examples, references, and some discussions of reasons for choosing the rank-rank specification.} The resulting measures are crucial inputs to broader political and public debates about inequality, about the importance of the family and the neighborhood into which children are born, and about how to create opportunities for children to rise out of poverty (\cite{Mogstad:2023uu}). Beyond intergenerational mobility, regressions involving ranks are used in a range of other areas such as behavioral economics (e.g., \cite{Huffman:2022lk}), development (e.g., \cite{Badge:2016yt}), education (e.g., \cite{Murphy:2020ii}), health (e.g., \cite{Gronqvist:2020oi}), industrial organization (e.g., \cite{Cage:2019oi}), labor (e.g., \cite{Faia:2023oi}), migration (e.g., \cite{Ward:2022jh}), and urban economics (e.g., \cite{Lee:2017uu}).

In its simplest form, a rank-rank regression in the intergenerational mobility literature consists of performing two steps: first, rank a child's and their parent's socioeconomic status (say, income) in their respective distributions and, second, run a regression of the child's rank on a constant and the parent's rank. The slope coefficient in this regression is then interpreted as a relative measure of intergenerational mobility. A small value of the slope indicates low dependence of the child's position in their income distribution on the parent's position in their income distribution, and thus high mobility. Applied work often interprets the slope coefficient as the rank correlation of the two incomes (e.g., \citet[p. 1555]{Chetty:2014tr} or \citet[p. 991]{Deutscher:2023oo}, which facilitates the interpretation as a parameter that takes values between minus one and one.

In this paper, we develop an asymptotic theory for the OLS estimator of such rank-rank regressions. First, we demonstrate that commonly used inference methods for rank-rank regressions are invalid. Specifically, we show that neither the homoskedastic nor the Eicker-White robust variance estimators consistently estimate the asymptotic variance of the OLS estimator. In fact, their probability limits may be too large or too small, depending on the shape of the copula of the two variables to be ranked. In consequence, inference based on these commonly used variance estimators may be conservative or fail to satisfy coverage criteria. The intuitive reason for this is that these variance estimators ignore the estimation error in the ranks. In addition, we show that, in the presence of pointmasses in the underlying distribution, the regression estimand, the OLS estimator and its asymptotic distribution may be sensitive to the way in which ties in the ranks are handled. In particular, the slope coefficient may not be equal to the rank correlation.

Motivated by these findings, we derive the asymptotic distribution of the OLS estimator in a general rank-rank regression with covariates and without assumptions on whether the underlying distribution is continuous or not. In the special case in which the underlying distribution is continuous and there are no covariates, the OLS estimator is equal to Spearman's rank correlation, and our limiting distribution coincides with that derived by \cite{Hoeffding:1948uu}. In general, however, the OLS estimator is not equal to Spearman's rank correlation and \cite{Hoeffding:1948uu}'s results do not apply. In particular, in the presence of pointmasses, our result shows how the regression estimand, the OLS estimator, and the asymptotic variance all depend on the definition of the ranks, i.e. how ties in the ranks are handled. Inference on the regression coefficients can be based either on the plugin estimator of the asymptotic variance\footnote{A software implementation is provided in the R package \texttt{csranks}, available on CRAN.}  or on the nonparametric bootstrap.

Finally, in two empirical applications, we illustrate the importance of employing valid inference methods in rank-rank regressions and the sensitivity of results to how ties in the ranks are handled.

Our paper contributes to the general literature on nonparametric rank statistics. Since the literature is large, we provide here only some key references, referring an interested reader, for example, to a recent review \cite{C23}. As mentioned above, \cite{Hoeffding:1948uu} derived the asymptotic distribution for Spearman's rank correlation in the case of continuous distributions. \cite{N07} defined rank correlation measures for noncontinuous distributions, using a specific way of handling ties, with the extension satisfying several intuitive axioms that any measure of concordance between random variables should satisfy. \cite{MQ10} extended the results of \cite{N07} to cover the case of more than two variables. \cite{GNR13} studied the problem of estimating rank correlations in the multivariate case with discontinuities in the distribution. \cite{OL16} derived the asymptotic normality result and the asymptotic variance formula for Spearman's rank correlation for variables with finite support. We note here that these extensions are, although relevant, quite different from our work. In particular, we study the OLS estimator in general regression specifications involving ranks, which coincide with Spearman's rank correlation only in the special case in which the marginal distributions are continuous and there are no covariates. Otherwise, our estimands and estimators differ from those in the references listed here. 

Our work is also related to \cite{Klein:2020oi}, \cite{Mogstad:2024aa}, and \cite{Bazylik:2025uu} but the key difference is that they focus on inference for different objects of interest. These papers consider a set of populations (e.g., countries) that are ranked by some feature of their distributions (e.g., average performance of students within a country). They then provide methods for inference on the ranks of the populations. In contrast, in our paper, we do not perform inference directly on ranks, but instead first rank observations and then perform inference on features of the distribution of these ranked observations.

\subsection{Survey on the Use of Regressions Involving Ranks}
\label{sec: survey}

Empirical researchers often transform variables into ranks before using them in a regression and then employ standard methods for inference that do not account for the estimation error in the ranks. To document this practice, we used Google Scholar to search for articles (and their appendices) published between January 2013 and February 2024 containing the words ``rank'' and ``regression''. We restricted the search to the following journals, which publish predominantly applied papers: American Economic Review (excluding comments, P\&P), Journal of Political Economy (excluding JPE Micro and JPE Macro), Quarterly Journal of Economics, and Review of Economic Studies. We then systematically went through all results that Google Scholar returned, dropping articles without an empirical application and articles that used the word ``rank'' in a different context. This process led to a sample of 62 articles.

Many of the articles contain a large number of different regression specifications involving ranks, but we only record which specifications (``rank-rank'', ``level-rank'', ``rank-level''), which types of estimators (``OLS'', ``TSLS'', ``nonparametric'', ``other''), and which types of standard errors (``homoskedastic/unknown'', ``robust'', ``clustered'', ``bootstrap'', ``other'', ``none'') occur in a paper. ``Rank-rank'' refers to a regression in which both the outcome and at least one regressor have been transformed into ranks before running the regression. Similarly, ``level-rank'' (``rank-level'') refers to a regression in which at least one regressor (only the outcome) but not the outcome (none of the regressors) has been transformed into ranks. We obtain a total of 153 regressions, i.e. (paper $\times$ regression specification $\times$ estimator type $\times$ standard error type) combinations, in the sample.

\begin{table*}\centering
\begin{tabular}{lrrr}
\toprule
 & level-rank & rank-level & rank-rank\\
\midrule
total & 83 & 30 & 40\\
\\  \multicolumn{4}{l}{{\it Panel A: by type of estimator}}\\
OLS & 33 & 19 & 20\\
TSLS & 4 & 3 & 1\\
nonparametric & 34 & 5 & 18\\
other & 12 & 3 & 1\\
\\ \multicolumn{4}{l}{{\it Panel B: by type of standard error}}\\
homoskedastic/unknown & 13 & 4 & 13\\
robust & 9 & 6 & 0\\
clustered & 23 & 9 & 5\\
bootstrap & 0 & 1 & 1\\
other & 0 & 1 & 0\\
none & 38 & 9 & 21\\
\bottomrule
\end{tabular}
\caption{Number of regressions by specification, type of estimator, and type of standard error. The regressions are from a sample of articles published between January 2013 and February 2024 in American Economic Review, Journal of Political Economy, Quarterly Journal of Economics, and Review of Economic Studies.}
\label{tab: lit summary stats}
\end{table*}

Table~\ref{tab: lit summary stats} categorizes the 153 regressions by specification (in the different columns), by type of estimator (Panel A), and by type of standard error (Panel B). First, consider the 40 regressions involving rank-rank specifications. Half of them (20) were estimated by OLS. For many regressions (11), it was not specified which method for the computation of standard errors was used. In most of these, we suspected the use of homoskedastic variance estimators and thus grouped them together with the one paper (2 regressions) that explicitly indicated ``homoskedastic'' standard errors. Many regression results were presented without standard errors (21). Not a single paper employed robust standard errors.

Level-rank specifications occurred in 83 of the regressions, most of them estimated by OLS (33) or nonparametric methods (34). For these regressions, clustering or not reporting standard errors was common. Rank-level specifications occurred somewhat less frequently in our sample (30).

To gain some insight into the topics of the articles we calculated the number of times each top-level JEL-code appeared on the articles in our sample. Labor and Demographic Economics (JEL codes starting with ``J'') were the most frequent, but health, education, public economics, and microeconomics were also listed as classifications by a substantial share of the articles. Overall, the results indicate that regressions involving ranks are used in a variety of subfields of economics.

\section{Motivation}
\label{sec: motivation}

To motivate the subsequent theoretical developments, in this section we first show that commonly used inference methods for rank-rank regressions are not valid, even in the simplest case when the underlying variables are continuously distributed and there are no covariates. We then explain how the presence of pointmasses or covariates affects the estimand and the statistical properties of the OLS estimator.

\subsection{Rank-Rank Regressions}\label{sub: notations rank-rank}
For concreteness, consider a child's income $Y$ and their parent's income $X$. Suppose we are interested in measuring intergenerational mobility by running a regression of the child's income rank on a constant and the parent's income rank. In this section, we assume there are no further covariates. The slope coefficient from this regression reflects the persistence of the two generations' positions in their respective income distributions. A small value of the slope indicates low persistence and thus high mobility.

Let $F$ be the joint distribution of the two random variables $X$ and $Y$, and let $F_X$ and $F_Y$ be the corresponding marginals. Suppose we observe an i.i.d. sample $\{(X_i,Y_i)\}_{i=1}^n$ from $F$. The rank-rank regression involves ranks of $X_i$ and $Y_i$. To deal with potential ties, we consider a general definition of the rank: for a fixed, user-specified $\omega\in[0,1]$, let 
$$R_X(x) := \omega F_X(x) + (1-\omega)F_X^-(x),\quad x\in\mathbb R, $$
where $F_X^-(x):=P(X<x)$. We then define the rank of $X_i$ as $R_i^X:=\hat{R}_X(X_i)$, where
\begin{equation}\label{eq: def rank}
    \hat{R}_X(x) := \omega\hat{F}_X(x) + (1-\omega)\hat{F}_X^-(x)+\frac{1-\omega}{n},\quad x\in\mathbb R,
\end{equation}
is an estimator of $R_X(x)$ and $\hat F_X^-(x) = n^{-1}\sum_{i=1}^n \ind\{X_i < x\}$ is an estimator of $F_X^-(x)$. This definition of the rank is such that a large value of $X_i$ is assigned a large rank. If $F_X$ is continuous, then $F_X$ and $F_X^-$ coincide and the population rank $R_X(X)$ is independent of $\omega$. Similarly, if there are no ties in the sample, then the sample rank $R_i^X$ is independent of $\omega$. If $F_X$ is not continuous, then different choices of $\omega$ lead to definitions of the rank that handle ties differently. For instance, if $\omega=1$ then $R_i^X=\hat{F}_X(X_i)$ and tied observations are assigned the largest possible rank. Similarly, $\omega=0$ leads to a definition of the rank that assigns to tied observations the smallest possible rank. If $\omega=1/2$, then $R_i^X$ is the mid-rank as defined in \cite{Hoeffding:1948uu}, which assigns to tied observations the average of the smallest and largest possible ranks. Table~\ref{tab: example ranks} illustrates the different definitions of ranks in an example with ties. The quantities $F_Y$, $F_Y^-$, $R_Y$, and their estimators are defined analogously for $Y$. We assume that $\omega$ is the same value in $R_X$ and $R_Y$, so that ranks are defined consistently for both variables $X$ and $Y$. Throughout the paper, this value of $\omega$ is fixed and chosen by the researcher, so it does not appear as argument or index anywhere.

\begin{table*}\centering
    \begin{tabular}{l|cccccccccc}
    \toprule
    $i$                                    & 1   & 2   & {\bf 3  } & {\bf 4  } & 5   & 6   & {\bf 7  } & {\bf 8  } & {\bf 9  } & {\bf 10 }\\
    $X_i$                                  & 3   & 4   & {\bf 7  } & {\bf 7  } & 10  & 11  & {\bf 15 } & {\bf 15  }& {\bf 15 } & {\bf 15 }\\
    \midrule
    smallest rank: $R_i^X$ for $\omega=0$    & 0.1 & 0.2 & {\bf 0.3} & {\bf 0.3} & 0.5 & 0.6 & {\bf 0.7} & {\bf 0.7} & {\bf 0.7} & {\bf 0.7} \\
    mid-rank: $R_i^X$ for $\omega=0.5$     & 0.1 & 0.2 & {\bf 0.35} & {\bf 0.35} & 0.5 & 0.6 & {\bf 0.85} & {\bf 0.85} & {\bf 0.85} & {\bf 0.85} \\
    largest rank: $R_i^X$ for $\omega=1$     & 0.1 & 0.2 & {\bf 0.4} & {\bf 0.4} & 0.5 & 0.6 & {\bf 1  } & {\bf 1  } & {\bf 1  } & {\bf 1  }\\
    \bottomrule
    \end{tabular}
    \caption{Example of three different definitions of ranks}
    \label{tab: example ranks}
\end{table*}

In a rank-rank regression, we first compute the parent's income rank $R_i^X$ and the child's income rank $R_i^Y$. Then, we run a regression of $R_i^Y$ on a constant and $R_i^X$. The OLS estimator of the slope parameter can then be written as the sample covariance of $R_i^Y$ and $R_i^X$, $S_{YX}:=\frac{1}{n}\sum_{i=1}^n(R_i^Y-\bar{R}^Y)(R_i^X-\bar{R}^X)$ where $\bar{R}^X$ and $\bar{R}^Y$ are sample averages of $R_i^X$ and $R_i^Y$, divided by the sample variance of $R_i^X$, $S_X^2 := \frac{1}{n} \sum_{i=1}^n (R_i^X-\bar{R}^X)^2$:
\begin{equation*}
\hat{\rho} := \frac{S_{YX}}{S_X^2}.
\end{equation*}
Researchers then perform inference by computing standard errors and confidence intervals using the usual variance formulas for the OLS estimator. In particular, these could be the homoskedastic or the Eicker–White robust variance formulas:
$$
\hat{\sigma}^2_{hom} := \frac{1}{nS_X^2}\sum_{i=1}^n \hat{\varepsilon}_i^2 \quad\text{and}\quad \hat{\sigma}^2_{EW} := \frac{1}{nS_X^4}\sum_{i=1}^n \hat{\varepsilon}_i^2(R_i^X-\bar{R}^X)^2,
$$
where $\hat{\varepsilon}_i := R^Y_i - \bar{R}^Y -\hat{\rho}(R^X_i-\bar{R}^X)$, $i=1,\dots,n$, is the residual from the regression. The estimator $\hat\sigma_{hom}^2$ is suitable when regression errors are believed to be homoskedastic whereas $\hat\sigma_{EW}^2$ is robust to heteroskedasticity. However, as we now demonstrate, these formulas yield inconsistent estimators of the asymptotic variance of $\hat\rho$.

\subsection{Inconsistency of OLS Variance Formulas}
\label{sec: motivation continuous}

To simplify the discussion, suppose that $F_X$ and $F_Y$ are both continuous, so that with probability one there are no ties in the sample. In this case, the OLS estimator is equal to Spearman's rank correlation, which is defined as $\hat{\rho}_S := \frac{S_{YX}}{S_XS_Y}$ with $S_Y^2$ denoting the sample variance of $R_i^Y$:
\begin{equation}\label{eq: OLS estimator}
  \hat{\rho}  = \hat\rho_S \frac{S_Y}{S_X} = \hat{\rho}_S,
\end{equation}
with probability one.\footnote{To see this, note that by continuity of the distributions, with probability one there are no ties in the data and $R_1^Y,\dots,R_n^Y$ is just a reordering of $R_1^X,\dots,R_n^X$. Therefore, $S_X=S_Y$ with probability one.} Thus, the OLS estimator of the rank-rank regression and Spearman's rank correlation are numerically identical in finite samples. As a result, they must have the same asymptotic properties, which we now discuss. 

First, it is well-known (\cite{Hoeffding:1948uu}) that the probability limit of $\hat{\rho}_S$ is
\begin{equation*}
    \rho = 12 \int_0^1\int_0^1 C(u,v)dudv-3 = Corr(F_X(X),F_Y(Y)),
\end{equation*}
with $C$ denoting the copula of the pair $(X,Y)$, i.e. $C(x,y) := P(F_X(X)\leq x,F_Y(Y)\leq y)$ for all $x,y\in[0,1]$. Second, \cite{Hoeffding:1948uu} showed that
$$\sqrt{n}(\hat{\rho}_S-\rho) \to_D N(0,\sigma^2), $$
where
\begin{equation}\label{eq: avar spearman}
  \sigma^2 := 9 Var\Big((2F_X(X)-1)(2F_Y(Y)-1) + 4\psi_X(X) + 4\psi_Y(Y) \Big)
\end{equation}
and
\begin{align}
    \psi_X(x) &:= \int (F(x,y) - F_X(x)F_Y(y))dF_Y(y)\label{eq: psi x def},\quad x\in\mathbb R,\\
    \psi_Y(y) &:= \int (F(x,y) - F_X(x)F_Y(y))dF_X(x)\label{eq: psi y def},\quad y\in\mathbb R.
\end{align}

On the other hand, the probability limits of the variance estimators $\hat{\sigma}^2_{hom}$ and $\hat{\sigma}^2_{EW}$ are given by the following lemma.
\begin{lemma}\label{lem: plims}
    Let $\{(X_i,Y_i)\}_{i=1}^n$ be an i.i.d. sample from $F$ with continuous marginal distributions $F_X$ and $F_Y$. Then
    \begin{equation}\label{eq: hom plim}
       \sigma^2_{hom} := \plim_{n\to\infty}\hat{\sigma}^2_{hom} = 1-\rho^2
    \end{equation}
    and, letting $M_{kl} := E[(F_X(X)-1/2)^k(F_Y(Y)-1/2)^l]$, $k,l\in\{1,2,3\}$,
    \begin{equation}\label{eq: EW plim}
        \sigma_{EW}^2 := \plim_{n\to\infty}\hat{\sigma}_{EW}^2 = 144\left( M_{22} -2\rho M_{31} + \frac{\rho^2}{80}\right).
    \end{equation}
\end{lemma}

When $X$ and $Y$ are independent, then $\rho=0$ and both probability limits are equal to the correct (asymptotic) variance, $\sigma^2_{hom}=\sigma_{EW}^2 = \sigma^2=1$. In general, however, the three variances are different.  In fact, $\sigma_{hom}^2$ and $\sigma_{EW}^2$ can each be larger or smaller than the correct variance $\sigma^2$. To see this note first that all three variances depend on the joint distribution $F$ only through its copula $C$. Therefore, whether and by how much the variances differ from each other depends only on the shape of the copula of child and parent incomes. As an illustration, we thus compare the three variances in three simple parametric families of copulas: the Gaussian, the Student-t with one degree of freedom, and a quadratic copula.\footnote{\label{foot:copulas}The Gaussian and Student-t families are each indexed by a correlation parameter $\theta\in[-1,1]$. The quadratic copula with parameter $\theta\in[0,1]$ is defined as the copula of $(Y,X)$ where $X\sim U[-1/2,1/2]$, $Y = 1/2 + \theta X + (1-\theta)X^2 + \varepsilon$, and $\varepsilon \sim N(0,10^{-6})$.}

\begin{figure}[!t]
  \centering
  \includegraphics[width=\textwidth]{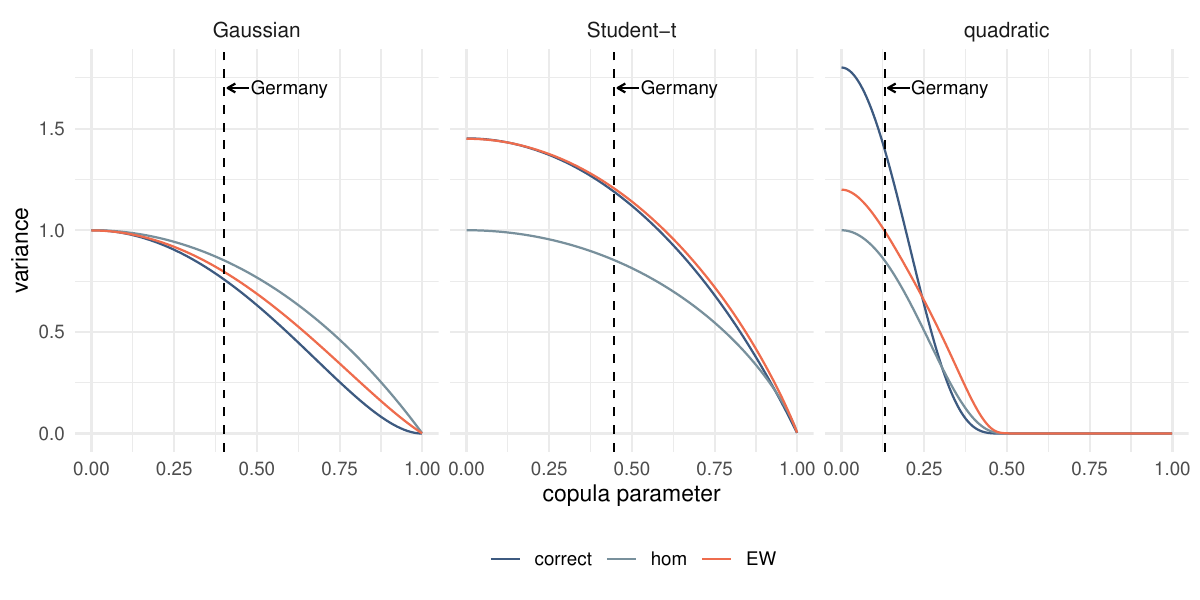}
  \caption{Variances within three different parametric families of copulas, plotted as functions of the copula parameter. ``correct'' refers to $\sigma^2$, ``hom'' to $\sigma_{hom}^2$, ``EW'' to $\sigma_{EW}^2$. The dashed line shows the value of the copula parameter that corresponds to the rank correlation of child and parent incomes in Germany ($0.384$), computed from the SOEP data used in Section~\ref{sec: SOEP}.}
  \label{fig: vars}
\end{figure}

Figure~\ref{fig: vars} plots the three variances within each copula family as we vary the copula parameter in $[0,1]$.\footnote{For each parametric family, the copula parameter refers to the parameter $\theta$ defined in Footnote~\ref{foot:copulas}.} The dashed line shows the value of the copula parameter value corresponding to the rank correlation between child and parent incomes in Germany ($0.384$), computed from the SOEP data used in Section~\ref{sec: SOEP}. For Germany, the quadratic copula produces hom and EW variances that are substantially smaller than the correct variance. The Student-t copula also yields a hom variance well below the correct variance, but the corresponding EW variance is essentially equal to the correct variance. The Gaussian copula generates hom and EW variances that are close to the correct variance, though both are slightly larger than the correct variance over the interior of the parameter space.

The variances $\sigma^2$, $\sigma_{hom}^2$, and $\sigma_{EW}^2$ are all bounded, and so mutual differences between them are bounded as well. This is not true for ratios, however. The following lemma shows that if we do not restrict the distributions $F$, the ratios $\sigma_{hom}^2/\sigma^2$ and $\sigma_{EW}^2/\sigma^2$ can both be arbitrarily large. To emphasize the dependence of the variances on the joint distribution $F$ of $(X,Y)$, we use the notation $\sigma_{hom}^2(F)$, $\sigma_{EW}^2(F)$, and $\sigma^2(F)$.

\begin{lemma}\label{lem: large ratio of variances}
There exists a sequence $\{F_k\}_{k\geq 1}$ of distributions on $\mathbb{R}^2$ with continuous marginals such that $\sigma_{hom}^2(F_k)/\sigma^2(F_k)\to\infty$ and $\sigma_{EW}^2(F_k)/\sigma^2(F_k)\to\infty$ as $k\to\infty$.
\end{lemma}

An implication of this lemma is that the two variances $\sigma_{hom}^2$ and $\sigma_{EW}^2$ can be larger than the correct variance by arbitrarily large factors. As the proof of the lemma reveals, the divergence occurs under sequences of copulas approaching perfect dependence so that both variances tend to zero, but the correct variance converges to zero at a faster rate than the  variances $\sigma_{hom}^2$ and $\sigma_{EW}^2$.

Given the result in Lemma \ref{lem: large ratio of variances}, the next interesting question is how small the variances $\sigma_{hom}^2$ and $\sigma_{EW}^2$ can be relative to the true variance $\sigma^2$. The following lemma provides a partial answer to this question in the case of $\sigma_{hom}^2$.

\begin{lemma}\label{lem: small ratio of variances}
    There exists a constant $c>0$ such that $\sigma_{hom}^2(F)/\sigma^2(F)\geq c$ for all distributions $F$ on $\mathbb{R}^2$ with continuous marginals.
\end{lemma}

By this lemma, the ratio $\sigma_{hom}^2/\sigma^2$ is bounded away from zero, so that $\sigma_{hom}^2$ can be smaller than $\sigma^2$ only by a factor that is bounded from below by the constant $c$.\footnote{We do not know whether a version of Lemma~\ref{lem: small ratio of variances} holds for the EW variance $\sigma_{EW}^2$ but we would be surprised if that were not the case.} However, it is important to emphasize that this constant may be rather small. In particular, this constant is strictly smaller than one, as there do exist distributions $F$ such that $\sigma_{hom}^2(F)/\sigma^2(F)<1$. For example, we have already seen in Figure~\ref{fig: vars} that for simple well-known classes of copulas, like the Student-t copula, the hom variance may be substantially smaller than the correct variance. In the empirical applications of Section~\ref{sec: emp}, we also find that the hom variance may be substantially smaller than the correct variance.

We conclude this subsection with a heuristic explanation for why the hom and EW variance estimators fail to consistently estimate the asymptotic variance of the OLS estimator. The OLS estimator can be written as
$$\hat\rho = \frac{S_{YX}}{S_X^2} = \frac{1}{n S_X^2} \sum_{i=1}^n (R_i^Y-\bar{R}^Y)(R_i^X-\bar{R}^X). $$
If $R_i^X$ and $R_i^Y$ were replaced by the population ranks $R_X(X_i)$ and $R_Y(Y_i)$, then, after scaling and recentering, the above quantity would be asymptotically equivalent to a sample average of i.i.d. random variables with asymptotic variance equal to $\sigma_{EW}^2$. However, since the ranks $R_i^X$ and $R_i^Y$ are estimators of the populations ranks, the above argument ignores the fact that these variables are sample averages themselves. Therefore, the OLS estimator is, in fact, a triple sum over the sample divided by the sample variance $S_X^2$ and asymptotically equivalent to a U-statistic of order three. Employing U-statistics theory from, e.g. \cite{Serfling:2002re}, one can show that the OLS estimator is asymptotically normal, but the variance is $\sigma^2$, not $\sigma_{EW}^2$. Intuitively, the estimation error in the ranks, i.e. the estimation error in the empirical cdfs, cannot be ignored even in the limit as the sample size tends to infinity because it is of the same order of magnitude as the estimation error in the infeasible OLS estimator in which the estimated ranks are replaced by the population ranks.

In conclusion, this subsection has shown that commonly used estimators for the asymptotic variance of the OLS estimator in a rank-rank regression are not consistent. In particular, the resulting homoskedastic and Eicker-White standard errors may be too small or too large depending on the shape of the underlying copula. In consequence, confidence intervals may be too short or too wide, possibly leading to under-coverage or conservative inference.

\subsection{Noncontinuous Marginal Distributions}
\label{sec: noncon dist}

The well-established asymptotic theory for Spearman's rank correlation described in the previous subsection requires both marginal distributions, $F_X$ and $F_Y$, to be continuous. In empirical applications of rank-rank regressions, however, pointmasses are common. For instance, incomes may be top-coded and there may be pointmasses at zero or negative incomes (e.g. as in \cite{Chetty:2018iu}).\footnote{In fact, one of the commonly cited (\cite{Deutscher:2023oo}, \cite{Mogstad:2023uu}) advantages of measuring mobility by the rank-rank slope is that it can be estimated even in the presence of zero incomes, while other popular measures like elasticities can only be estimated after removing observations with zero incomes.} Alternatively, ranks may be computed from discrete measures other than income, e.g., occupational status as in \cite{Ward:2023uu}, human capital as in \cite{Croix:2022aa}, or years of education as in \cite{Asher:2024iu}. In such cases, pointmasses in the marginal distributions create ties in the ranks. The presence of ties changes inference in the rank-rank regression in at least three important ways, which we now discuss.

\paragraph*{Interpretation of the estimand.} The probability limit of the OLS estimator is equal to 
\begin{equation}\label{eq: plim rhohat}
    \rho = Corr(R_Y(Y),R_X(X)) \sqrt{\frac{Var(R_Y(Y))}{Var(R_X(X))}},
\end{equation}
i.e. it is equal to the rank correlation times the ratio of standard deviations of the ranks. For continuous distributions, the ratio of standard deviations is equal to one. However, in the presence of pointmasses, this ratio can take any value in $(0,\infty)$ and, thus, the slope coefficient of the rank-rank regression may lie outside of the $[-1,1]$ interval and cannot be interpreted as a correlation.

\paragraph*{Sensitivity of the estimand to the definition of ranks.} In the presence of pointmasses, the population ranks $R_X(X)$ and $R_Y(Y)$ depend on $\omega$, i.e. on how ranks are assigned to tied observations. Therefore, the OLS estimand \eqref{eq: plim rhohat} also depends on $\omega$. In other words, different ways of handling ties imply different estimands.

\paragraph*{Sensitivity of the OLS estimator's statistical properties to the definition of ranks.} In the presence of pointmasses, not only the population ranks, but also the sample ranks $R_i^X$ and $R_i^Y$ depend on $\omega$. Therefore, the OLS estimator and its statistical properties also depend on $\omega$ and, thus, on how ties are handled. The previous paragraph already argues that the estimator's probability limit depends on $\omega$, but in Section~\ref{sec: asymptotic normality} we point out that its asymptotic variance also depends on $\omega$. In consequence, standard errors for the OLS estimator vary depending on how ties in the ranks are handled. In the empirical application of Section~\ref{sec: Asher et al}, we demonstrate that this sensitivity can be so severe that conclusions are qualitatively different depending on how ties are handled.

\subsection{Covariates}

As seen in \eqref{eq: plim rhohat}, the rank-rank slope $\rho$ is not necessarily equal to the rank correlation. This case occurs when the standard deviations of $R_Y(Y)$ and $R_X(X)$ differ due to the presence of pointmasses.

Another reason for why $\rho$ may not be interpretable as a correlation is the presence of covariates, say $W$, in the rank-rank regression. With covariates the rank-rank slope satisfies
$$\rho = Corr(\tilde{R}_Y(Y),\tilde{R}_X(X))  \sqrt{\frac{Var(\tilde{R}_Y(Y))}{Var(\tilde{R}_X(X))}},$$
where $\tilde{R}_Y(Y)$ and $\tilde{R}_X(X)$ denote the residuals from partialling out $W$ from $R_Y(Y)$ and $R_X(X)$. Therefore, the rank-rank slope is equal to the correlation of $\tilde{R}_Y(Y)$ and $\tilde{R}_X(X)$ only if the standard deviations of the two residuals happen to be identical. In general, however, the ratio of standard deviations can take any value in $(0,\infty)$ and, thus, the rank-rank slope may lie outside of $[-1,1]$.

\section{A General Asymptotic Theory}
\label{sec: inf rank rank reg}

\subsection{Asymptotic Normality Result}
\label{sec: asymptotic normality}

In this section, we develop a general, unifying asymptotic theory for coefficients in a rank-rank regression that allows for any definition of the rank, for continuous or noncontinuous distributions, and for the presence of covariates.

We consider the following regression model:
\begin{equation}\label{eq: model with controls}
R_Y(Y) = \rho R_X(X) + W'\beta + \varepsilon,\qquad E\left[\varepsilon\begin{pmatrix}
R_{X}(X)\\ W
\end{pmatrix}\right]=0,
\end{equation}
where $W$ is a $d$-dimensional vector of covariates, $\varepsilon$ is noise, and $\rho$, $\beta$ are regression coefficients to be estimated. In applications, $W$ typically includes a constant, but our theoretical results below do not require that.

Letting $\{(X_i,W_i,Y_i)\}_{i=1}^n$ be an i.i.d. sample from the distribution of the triplet $(X,W,Y)$, we study the properties of the OLS estimator of the parameters $(\rho,\beta')'$:
\begin{equation}\label{eq: joint ols estimator}
\begin{pmatrix}
\hat{\rho}\\
\hat{\beta}
\end{pmatrix}=\left(\sum_{i=1}^{n}\begin{pmatrix}
R_i^X\\
W_{i}
\end{pmatrix}\begin{pmatrix}
R_i^X & W_{i}'\end{pmatrix}\right)^{-1}\sum_{i=1}^{n}\begin{pmatrix}
R_i^X\\
W_{i}
\end{pmatrix}R_i^Y.
\end{equation}
The population ranks $R_X(X)$, $R_Y(Y)$ and the estimated ranks $R_i^X$, $R_i^Y$ are defined as in Section \ref{sub: notations rank-rank} with the same, but arbitrary value of $\omega\in[0,1]$.

To derive the asymptotic normality of $\hat\rho$, we introduce the projection of $R_X(X)$ onto the covariates:
\begin{equation}\label{eq: first stage}
    R_X(X) = W'\gamma + \nu,\quad E[\nu W] = 0,
\end{equation}
where $\nu$ is a random variable representing the projection residual, and $\gamma$ is a vector of parameters. Consider the following regularity conditions:
\begin{assumption}\label{as: random sample}
    $\{(X_i,W_i,Y_i)\}_{i=1}^n$ is an i.i.d. sample from the distribution of $(X,W,Y)$.
\end{assumption}

\begin{assumption}\label{as: vector w}
    The vector $W$ satisfies $E[\|W\|^4]<\infty$ and the matrix $E[WW']$ is non-singular.
\end{assumption}
\begin{assumption}\label{as: variable nu}
    The random variable $\nu$ is such that $E[\nu^2]>0$.
\end{assumption}
These are standard regularity conditions underlying typical regression analyses. Assumption~\ref{as: variable nu} requires that the rank $R_X(X)$ can not be represented as a linear combination of covariates. Importantly, our regularity conditions do not require $X$ and $Y$ to be continuously distributed.

Under these conditions, we have the following asymptotic normality result.

\begin{theorem}\label{thm: AN for rank rank reg}
    Suppose that \eqref{eq: model with controls}--\eqref{eq: first stage} hold and that Assumptions \ref{as: random sample}--\ref{as: variable nu} are satisfied. Then
    $$
    \sqrt n(\hat \rho - \rho) = \frac{1}{\sigma_{\nu}^2\sqrt n}\sum_{i=1}^n \Big\{ h_1(X_i,W_i,Y_i) + h_2(X_i,Y_i) + h_3(X_i)\Big\} + o_P(1)\to_D N(0,\sigma^2),
    $$
    where 
    \begin{equation}\label{eq: our general variance formula}
    \sigma^2  := \frac{1}{\sigma_{\nu}^4} E\Big[(h_1(X,W,Y) + h_2(X,Y) + h_3(X))^2\Big] 
    \end{equation}
    with $\sigma_{\nu}^2 := E[\nu^2]$, $I(u,v) := \omega \ind\{u\leq v\} + (1-\omega)\ind\{u<v\}$, and
    \begin{align*}
        h_1(x,w,y) &:= (R_Y(y) - \rho R_X(x) - w'\beta)(R_X(x) - w'\gamma),\\
        h_2(x,y) & := E[(I(y,Y) - \rho I(x,X) - W'\beta)(R_X(X) - W'\gamma)],\\
        h_3(x) & := E[(R_Y(Y) - \rho R_X(X) - W'\beta)(I(x,X) - W'\gamma)]
    \end{align*}
for all $x\in\mathbb R$, $w\in\mathbb R^p$, and $y\in\mathbb R$.
\end{theorem}
The proof of this result can be found in Appendix~\ref{sec: proof main result}. In fact, Theorem \ref{thm: AN for rank rank reg} follows from a more general result, which we state in Appendix~\ref{sec: all coefficients} (Theorem \ref{thm: AN for rank rank reg all coeffs}). The more general result shows joint asymptotic normality for $\hat\rho$ and $\hat\beta$,
$$\sqrt{n} \begin{pmatrix}
    \hat\rho -\rho\\ \hat\beta-\beta 
\end{pmatrix}= \frac{1}{\sqrt{n}} \sum_{i=1}^n \psi_i + o_P(1) \to_D N(0,\Sigma),\qquad \Sigma := E[\psi_i\psi_i'], $$
and provides an explicit formula for $\psi_i$. The joint asymptotic distribution for $\hat\rho$ and $\hat\beta$ is useful for empirical work in the intergenerational mobility literature. Suppose there are no covariates so $W$ only contains a constant. Besides the rank-rank slope, another popular measure of intergenerational mobility (\cite{Deutscher:2023oo}) is the expected income rank of a child given that their parent's income rank is equal to some value $p$: $\theta_p := \beta + \rho\, p$. With the asymptotic joint distribution above, it is straightforward to compute standard errors and confidence intervals for $\theta_p$ based on the plugin estimator $\hat{\theta}_p := \hat{\beta} + \hat{\rho}\, p$.

Theorem \ref{thm: AN for rank rank reg} has the following interpretation. If we knew the population ranks $R_Y(Y_i)$ and $R_X(X_i)$, we would simply estimate $\rho$ and $\beta$ by an OLS regression of $R_Y(Y_i)$ on $R_X(X_i)$ and $W$. The asymptotic variance of such an estimator for $\rho$ would be given by $E[h_1(X,W,Y)^2]/\sigma_{\nu}^4$, which is the familiar OLS variance for heteroskedastic regression errors. Without knowing the population ranks, however, we have to plugin their estimators, and the extra terms in the asymptotic variance, represented by $h_2(X,Y)$ and $h_3(X)$, provide the adjustments necessary to take into account the noise coming from estimated ranks. Importantly, the estimation error in the ranks is of the same order of magnitude as that in the infeasible OLS estimator in which the estimated ranks are replaced by the population ranks. Therefore, the estimation error in the ranks does not become negligible in large samples, not even in the limit as the sample size tends to infinity.

\begin{remark}[Comparing the variance formulas in \eqref{eq: avar spearman} and \eqref{eq: our general variance formula}]
When both $X$ and $Y$ are continuous random variables and $W$ contains only a constant, the asymptotic variance formula in \eqref{eq: our general variance formula} reduces to the classical Hoeffding variance formula in \eqref{eq: avar spearman}. Indeed, replacing $X$ and $Y$ by $F_X(X)$ and $F_Y(Y)$ respectively, we can assume without loss of generality that both $X$ and $Y$ are $U[0,1]$ random variables, in which case $\psi_X(x)$ in \eqref{eq: psi x def} reduces to
\begin{align*}
\int_0^1 (E[\ind\{X\leq x\}\ind\{Y\leq y\}] - xy)dy 
& = E[\ind\{X\leq x\}(1-Y)] - x/2 \\
& = x/2 - E[\ind\{X\leq x\}Y].
\end{align*}
Similarly, $\psi_Y(y)$ in \eqref{eq: psi y def} reduces to $y/2 - E[\ind\{Y\leq y\}X]$. Hence, the variance in \eqref{eq: avar spearman} simplifies to $\sigma^2 = 144Var(h(X,Y))$, where
\begin{equation}\label{eq: function h true}
h(x,y):=xy - E[\ind\{X\leq x\}Y] - E[\ind\{Y\leq y\}X],\quad x\in\mathbb R, y\in\mathbb R,
\end{equation}
and it is straightforward to verify that $h(x,y)$ coincides with $h_1(x,w,y) + h_2(x,y) + h_3(x)$ up to an additive constant whenever $X,Y\sim U[0,1]$ and $W$ contains only a constant. \QEDA
\end{remark}

\begin{remark}[pointmasses]\label{rem: noncontinuous dist}
The asymptotic normality result in Theorem~\ref{thm: AN for rank rank reg} holds for both continuous and noncontinuous distributions of $Y$ and $X$. In the presence of pointmasses in at least one of the distributions of $Y$ and $X$, the estimands $\rho$ and $\beta$ depend on $\omega$, i.e. on the way in which ties are handled. Similarly, the asymptotic variance $\sigma^2$ of the OLS estimator derived in Theorem~\ref{thm: AN for rank rank reg} depends on $\omega$. Therefore, different ways of handling ties through different choices of $\omega$ may affect not only the estimand and the estimator, but also the value of the asymptotic variance.
\QEDA
\end{remark}

\begin{remark}[choice of $\omega$]
    In the presence of pointmasses in the marginal distributions of $Y$ and $X$, the researcher needs to choose the value of $\omega\in[0,1]$ that defines how ties are handled. In Section~\ref{sec: Asher et al}, we show that results might be highly sensitive to $\omega$. Therefore, a natural question is how to best choose this parameter. First, notice that each value of $\omega$ may induce a different estimand (i.e., $\rho$ depends on $\omega$). Because of that one might argue $\omega$ should be chosen by the researcher prior to looking at the data so as to define the estimand of interest. On the other hand, we cannot imagine substantive reasons for one of these estimands to be more suitable than others. Therefore, one might also consider various strategies of avoiding an arbitrary choice of $\omega$. For instance, one could compute the OLS estimator over a grid of values of $\omega$ and then report the average of the estimators. One could also choose the value of $\omega$ so as to minimize the OLS estimator's asymptotic variance. Finally, one could consider a randomized definition of the ranks that randomly breaks ties.\QEDA
\end{remark}

\subsection{Consistent Estimation of the Asymptotic Variance}
\label{sec: const var estim}

In this subsection, we propose an estimator of the asymptotic variance $\sigma^2$ appearing in the asymptotic normality result in Theorem \ref{thm: AN for rank rank reg} and show that it is consistent. In particular, we consider the following plug-in estimator:
$$\hat{\sigma}^2 := \frac{1}{n \hat{\sigma}_{\nu}^4} \sum_{i=1}^n(H_{1i} + H_{2i} + H_{3i})^2, $$
where $\hat{\sigma}_{\nu}^2 := n^{-1}\sum_{i=1}^n \hat{\nu}_i^2$ is an empirical analog of $\sigma_{\nu}^2 = E[\nu^2]$, $\hat{\nu}_i := R_i^X-W_i'\hat{\gamma}$, and
\begin{align*}
    H_{1i} &:= \left(R_i^Y - \hat{\rho} R_i^X - W_i'\hat{\beta}\right)\left(R_i^X - W_i'\hat{\gamma}\right),\\
    H_{2i} & := \frac{1}{n}\sum_{j=1}^n \left(I(Y_i,Y_j) - \hat{\rho} I(X_i,X_j)-W_j'\hat{\beta}\right)\left(R_j^X-W_j'\hat{\gamma}\right),\\
    H_{3i} & := \frac{1}{n}\sum_{j=1}^n \left(R_j^Y - \hat{\rho} R_j^X - W_j'\hat{\beta}\right)\left(I(X_i,X_j) - W_j'\hat{\gamma}\right),
\end{align*}
for all $i=1,\dots,n$. The following lemma shows that this simple plug-in estimator is consistent without any additional assumptions.

\begin{lemma}\label{lem: consistent variance estimation}
Suppose that \eqref{eq: model with controls}--\eqref{eq: first stage} hold and that Assumptions \ref{as: random sample}--\ref{as: variable nu} are satisfied. Then $\hat\sigma^2\to_P\sigma^2$.
\end{lemma}

Theorem \ref{thm: AN for rank rank reg} and Lemma \ref{lem: consistent variance estimation} give the correct way to perform inference in rank-rank regressions. For example, a $(1-\alpha)\times 100\%$ asymptotic confidence interval for $\rho$ can be constructed using the standard formula
$$
\left(\hat\rho - \frac{z_{\alpha/2}\hat\sigma}{\sqrt n}, \hat\rho + \frac{z_{\alpha/2}\hat\sigma}{\sqrt n}\right),
$$
where $z_{\alpha/2}$ is the number such that $P(N(0,1) > z_{\alpha/2}) = \alpha/2$. Standard hypothesis testing can be performed analogously.

\begin{remark}[bootstrap]\label{rem: bootstrap}
    In addition to performing inference on $\rho$ via the plugin estimator of the asymptotic variance, one can also perform inference by bootstrapping the distribution of $\sqrt n(\hat\rho - \rho)$. Indeed, consider a bootstrap version $\hat\rho^*$ of the estimator $\hat\rho$ constructed in three steps: (1) draw a bootstrap sample from the original sample $\{(X_i,W_i,Y_i)\}_{i=1}^n$ with replacement, (2) calculate bootstrap ranks, i.e. ranks on the bootstrap sample, and (3) calculate the OLS estimator $\hat\rho^*$ using the bootstrap ranks. One can show that $\rho$ is a Hadamard differentiable functional of the distribution of the data, regardless of whether that distribution is continuous or not. Therefore, by standard results, e.g. from \cite{vaart}, the distribution of $\sqrt n(\hat\rho^* - \hat\rho)$ consistently estimates the asymptotic distribution of $\sqrt n(\hat\rho - \rho)$, again regardless of whether the distribution of the data is continuous or not. See Appendix~\ref{app: bootstrap}. 

    Even when the marginal distributions of $Y$ and $X$ are continuous, the bootstrap sample inevitably contains ties, but these do not lead to problems for bootstrap inference. In this case, the bootstrap distribution depends on how ties are handled (i.e., on $\omega$), but regardless of which value of $\omega$ is chosen, the bootstrap distribution converges to the correct limit, which is independent of $\omega$.
    
    Finally, it is important to emphasize that the validity of the bootstrap hinges on re-sampling the original data $\{(X_i,W_i,Y_i)\}_{i=1}^n$ and re-computing ranks on each bootstrap sample; it is not sufficient to draw bootstrap samples directly from $\{(R_i^X, W_i, R_i^Y)\}_{i=1}^n$. \QEDA
\end{remark}

\section{Other Regressions Involving Ranks}
\label{sec: other regressions}

Motivated by the empirical practice documented in Section~\ref{sec: survey} we now extend the asymptotic normality result in Theorem~\ref{thm: AN for rank rank reg} to other regressions involving ranks. We consider (i) a rank-rank regression with subpopulations, where ranks are computed in the national distributions rather than in the subpopulation-specific distributions, (ii) a regression of the level of $Y$ on the rank of $X$, and (iii) a regression of the rank of $Y$ on the level of $X$. For brevity of the paper, we keep the discussions of each extension relatively short.

\subsection{Rank-Rank Regressions With Subpopulations}
\label{sec: rank rank reg with subpops}

In this subsection, we consider a population (e.g., the U.S.) that is divided into $n_G$ subpopulations (e.g., commuting zones). We are interested in running rank-rank regressions separately within each subpopulation. The ranks, however, are computed in the distribution of the entire population (e.g., the U.S.). \cite{Chetty:2014tr} has been influential in promoting this kind of regression for the analysis of mobility across regions, where the scale of the mobility measure is fixed by the national distribution. The survey by \cite{Deutscher:2023oo} provides more examples of empirical work running such regressions, for instance \cite{Corak:2020io} and \cite{Acciari:2022uu}. In Section~\ref{sec: SOEP}, we apply the methods from this section to study income mobility across states in Germany.

Consider the model
\begin{equation}\label{eq: model with subpops}
    R_Y(Y) = \sum_{g=1}^{n_G} \ind\{G=g\}\left(\rho_g R_X(X) + W'\beta_g\right) + \varepsilon,\quad E\left[\left.\varepsilon\begin{pmatrix} R_{X}(X)\\ W\end{pmatrix} \right| G\right]=0\; \text{a.s.},
\end{equation}
where $G$ is an observed random variable taking values in $\{1,\ldots,n_G\}$ to indicate the subpopulation to which an individual belongs. The quadruple $(G,X,W,Y)$ has distribution $F$, and we continue to denote marginal distributions of $X$ and $Y$ by $F_X$ and $F_Y$. The quantities $F_X^-$, $F_Y^-$, $R_X(x)$, and $R_Y(y)$ are also as previously defined, so that $R_X(X)$, for instance, is the rank of $X$ in the entire population, not the rank within a subpopulation. So, in model \eqref{eq: model with subpops}, the coefficients $\rho_g$ and $\beta_g$ are subpopulation-specific, but the ranks $R_Y(Y)$ and $R_X(X)$ are not. In consequence, $\rho_g$ cannot be interpreted as the rank correlation within the subpopulation $g$. Instead, in the intergenerational mobility literature, the rank-rank slope is interpreted as a relative measure of mobility in a region $g$, where its scale is fixed by the national population. Unlike the rank-rank slope in the model without subpopulations, \eqref{eq: model with controls}, the slopes $\rho_g$ do not only depend on the copula of $Y$ and $X$ in a subpopulation, but also on the marginal distributions of $Y$ and $X$ in the subpopulation. To see this note that adding a fixed amount to every child and parent income in a subpopulation does not change the ranking of children and parents within the subpopulation, but it may change the ranking of these individuals in their national income distributions. In conclusion, the $\rho_g$ may then also change.

We now introduce a first-stage projection equation similar to the one in \eqref{eq: first stage}, except that the coefficients are subpopulation-specific:
\begin{equation}\label{eq: first stage with subpops}
    R_X(X) = \sum_{g=1}^{n_G} \ind\{G=g\} W'\gamma_g + \nu,\quad E[\left.\nu W\right| G] = 0\; \text{a.s.}.
\end{equation}
Let $\{(G_i,X_i,W_i,Y_i)\}_{i=1}^n$ be an i.i.d. sample from the distribution of $(G,X,W,Y)$. Ranks are computed using all observations, i.e. $R_i^X := \hat R_X(X_i)$ with $\hat R_X(x)$ as in \eqref{eq: def rank} and $\hat{F}_X$ ($\hat{F}^-_X$) the (left-limit of the) empirical cdf of $X_1,\ldots,X_n$. The computation of the rank $R_i^Y$ is analogous.

First, notice that an OLS regression of $R_i^Y$ on all regressors, i.e. $(\ind\{G_i=g\}R_i^X)_{g=1}^{n_G}$ and $(\ind\{G_i=g\}W_i^X)_{g=1}^{n_G}$, produces estimates $(\hat\rho_g,\hat\beta_g)_{g=1}^{n_G}$ of $(\rho_g,\beta_g)_{g=1}^{n_G}$ that can be written as:
\begin{equation}\label{eq: joint ols estimator with subpops}
    \begin{pmatrix}
    \hat{\rho}_g\\
    \hat{\beta}_g
    \end{pmatrix}=\left(\sum_{i=1}^{n}\ind\{G_i=g\}\begin{pmatrix}
    R_i^X\\
    W_{i}
    \end{pmatrix}\begin{pmatrix}
    R_i^X & W_{i}'\end{pmatrix}\right)^{-1}\sum_{i=1}^{n}\ind\{G_i=g\}\begin{pmatrix}
    R_i^X\\
    W_{i}
    \end{pmatrix}R_i^Y.
\end{equation}
Therefore, $(\hat\rho_g,\hat\beta_g)$ can be computed by an OLS regression of $R_i^Y$ on $R_i^X$ and $W_i$ using only observations $i$ from subpopulation $g$. Similarly, $\hat{\gamma}_g$ can be computed by an OLS regression of $R_i^X$ on $W_i$ using only observations $i$ from subpopulation $g$. Note however, as explained above, that the ranks $R_i^X$ and $R_i^Y$ are computed using observations of $X$ and $Y$ from all subpopulations and thus the OLS estimators for different subpopulations are not independent.

The following are Assumptions~\ref{as: random sample}--\ref{as: variable nu} adapted to the model with subpopulations:

\begin{assumption}\label{as: random sample with subpops}
    $\{(G_i,X_i,W_i,Y_i)\}_{i=1}^n$ is a random sample from the distribution of $(G,X,W,Y)$.
\end{assumption}

\begin{assumption}\label{as: vector w with subpops}
    The vector $W$ is such that $E[\|W\|^4]<\infty$ and, for all $g=1,\ldots,n_G$, the matrix $E[\ind\{G=g\}WW']$ is non-singular.
\end{assumption}
\begin{assumption}\label{as: variable nu with subpops}
    The random variable $\nu$ is such that $E[\ind\{G=g\}\nu^2]>0$ for all $g=1,\dots,n_G$.
\end{assumption}

As in the previous section, note that our assumptions do not require the marginal distributions $F_X$ and $F_Y$ to be continuous. In addition, we introduce an assumption about the number and size of the subpopulations:
\begin{assumption}\label{as: large subpops}
    The number of subpopulations $n_G$ is a finite constant and $P(G=g)>0$ for all $g=1,\ldots,G$.
\end{assumption}

Observe that if the number of subpopulations were to increase together with the sample size $n$, with the number of units within each subpopulation being of the same order, the extra noise coming from estimated ranks would be asymptotically negligible, and the standard OLS variance formula would be applicable. The new result below applies to the case with a fixed number of subpopulations so that the estimation error in the ranks is not negligible even in large samples. This scenario seems reasonable, for instance, in our empirical application in which a subpopulation is a German state. 

Under these four assumptions, we have the following extension of Theorem \ref{thm: AN for rank rank reg}.

\begin{theorem}\label{thm: AN for rank rank reg with subpops}
    Suppose that \eqref{eq: model with subpops}--\eqref{eq: joint ols estimator with subpops} hold and that Assumptions \ref{as: random sample with subpops}--\ref{as: large subpops} are satisfied. Then for all $g=1,\dots,G$, we have
    \begin{align*}
        \sqrt n(\hat \rho_g- \rho_g) \to_D N(0,\sigma_g^2),
    \end{align*}
    for some $\sigma_g^2>0$.
 \end{theorem}

In Appendix~\ref{sec: all coefficients}, we prove a joint asymptotic normality result for all regression coefficients and provide explicit formulas for the asymptotic variance. In particular, letting $\hat\rho:=(\hat\rho_1,\ldots,\hat\rho_{n_G})'$, $\rho:=(\rho_1,\ldots,\rho_{n_G})'$, $\hat\beta:=(\hat\beta_1',\ldots,\hat\beta_{n_G}')'$, and $\beta:=(\beta_1',\ldots,\beta_{n_G}')'$, the appendix shows that
\begin{equation}\label{eq: AN all coeffs with subpops}
\sqrt{n} \begin{pmatrix}
    \hat\rho -\rho\\ \hat\beta-\beta 
\end{pmatrix}= \frac{1}{\sqrt{n}} \sum_{i=1}^n \psi_i + o_P(1) \to_D N(0,\Sigma),\qquad \Sigma := E[\psi_i\psi_i'], 
\end{equation}
and provides an explicit formula for $\psi_i$. From this result, one can then easily calculate the asymptotic distribution of linear combinations of parameters. For instance, similarly as in the rank-rank regression without subpopulations, a popular measure of intergenerational mobility (\cite{Deutscher:2023oo}) is the expected rank of a child with parents at a given income rank $p$,
\begin{equation}\label{eq: def theta}
    \theta_{g,p} := \beta_g + \rho_g\; p.
\end{equation}
The asymptotic distribution in \eqref{eq: AN all coeffs with subpops} allows us to construct a confidence interval for $\theta_{g,p}$ for a specific commuting zone $g$ or simultaneous confidence sets across all commuting zones. Here, an estimator of the asymptotic variance can be obtained using the plugin method and its consistency can be proved using the same arguments as those in Lemma \ref{lem: consistent variance estimation}.

\subsection{Level-Rank Regressions}
\label{sec: outcome on rank}
In this subsection, we consider a regression model with the level of $Y$ as the dependent variable and a rank as a regressor:
\begin{equation}\label{eq: model of outcome on rank}
    Y = \rho R_X(X) + W'\beta + \varepsilon,\qquad E\left[\varepsilon\begin{pmatrix}
    R_{X}(X)\\ W
    \end{pmatrix}\right]=0.
\end{equation}
Such a regression has been used, for instance, by \cite{Murphy:2020ii} who regress student outcomes like test scores on students' ranks in their classrooms. Other examples are \cite{Chetty:2014tr} and \cite{Abramitzky:2021ii} who regress a child's outcome like college attendance or teenage pregnancy on their parent's income rank.

For the above regression, the OLS estimator takes the form
\begin{equation}\label{eq: joint ols estimator outcome on rank}
    \begin{pmatrix}
    \hat{\rho}\\
    \hat{\beta}
    \end{pmatrix}=\left(\sum_{i=1}^{n}\begin{pmatrix}
    R_i^X\\
    W_{i}
    \end{pmatrix}\begin{pmatrix}
    R_i^X & W_{i}'\end{pmatrix}\right)^{-1}\sum_{i=1}^{n}\begin{pmatrix}
    R_i^X\\
    W_{i}
    \end{pmatrix}Y_i.
\end{equation}
The following theorem derives asymptotic normality for $\hat\rho$.

\begin{theorem}\label{thm: AN for outcome on rank reg}
    Suppose that \eqref{eq: first stage}, \eqref{eq: model of outcome on rank}, and \eqref{eq: joint ols estimator outcome on rank} hold, that $E[\varepsilon^4]<\infty$, and that Assumptions \ref{as: random sample}--\ref{as: variable nu} are satisfied. Then
    $$
    \sqrt n(\hat \rho - \rho) \to_D N(0,\sigma^2),
    $$
    for some $\sigma^2>0$.
\end{theorem}

In Appendix~\ref{sec: all coefficients}, we provide a joint asymptotic normality result for all regression coefficients and derive the expression for $\psi_i$. A consistent estimator of the asymptotic variance can be obtained by the plugin method, analogously to our discussion in Section~\ref{sec: const var estim}.

\subsection{Rank-Level Regressions}
\label{sec: rank on regressor}
In this subsection, we consider a regression model in which the outcome variable has been transformed into a rank, but the regressors are included in levels:
\begin{equation}\label{eq: model of rank on regressor}
    R_Y(Y) = W'\beta + \varepsilon,\qquad E\left[\varepsilon W \right]=0,
\end{equation}
where, for simplicity of notation, we let the vector $W$ absorb the regressor $X$. Examples of such a regressions appear in \cite{Ghosh:2023oi} who regress income rank on a binary indicator for whether a couple is in an isonymous marriage and \cite{Gronqvist:2020oi} who regress a student's grade point average rank on an indicator of childhood lead exposure.

The OLS estimator takes the following form:
\begin{equation}\label{eq: joint ols estimator rank on regressor}
    \hat{\beta} =\left(\sum_{i=1}^{n} W_{i}W_{i}'\right)^{-1}\sum_{i=1}^{n}W_iR^Y_i.
\end{equation}
We then have the following result.

\begin{theorem}\label{thm: AN for rank on regressor reg}
    Suppose that \eqref{eq: model of rank on regressor}--\eqref{eq: joint ols estimator rank on regressor} hold and that Assumptions \ref{as: random sample} and \ref{as: vector w} are satisfied. Then
    $$
    \sqrt n( \hat\beta - \beta) \to_D N(0,\Sigma),
    $$
    for some positive-definite $\Sigma$.
\end{theorem}
 
 An explicit formula for the matrix $\Sigma$ can be found in the beginning of the proof of this theorem and, like above, a consistent estimator of $\Sigma$ can be obtained by the plugin method, analogously to our discussion in Section~\ref{sec: const var estim}.

\section{Empirical Applications}
\label{sec: emp}

In two empirical applications, we illustrate the importance of employing valid inference methods in rank-rank regressions and the sensitivity of results to how ties in the ranks are handled.

\subsection{Intergenerational Income Mobility in Germany}
\label{sec: SOEP}

In this section, we study regional differences in intergenerational income mobility in Germany. Following \cite{Chetty:2014tr}, we measure mobility by rank-rank slopes and expected child ranks for the different states in Germany. As is common in the literature we then compare mobility in the different states to estimates from other countries. Such comparisons are important inputs to debates about potential causes for some regions or countries to have higher mobility than others. Through comparisons to high-mobility regions policy-makers in less mobile regions may seek to learn about potential policy levers for improving mobility.

\paragraph*{Data.} We use data from the German \cite{SOEP:2022aa} containing income on linked child-parent pairs. The sample is a subsample of that used in \cite{Dodin:2024op}, restricted to children for whom we observe the state in which they were born and restricted to states for which we observe at least 20 linked child-parent pairs. The resulting sample size is 664. The incomes of the child and the parent are computed as gross family income as in \citet[Footnote b in Table 4]{Dodin:2024op}. 

For the comparison with the U.S., we also use the publicly available estimates of expected child ranks and rank-rank slopes for each commuting zone (CZ) from \cite{Chetty:2014tr}, provided in the replication package \cite{Chetty:2022oi}.

\paragraph*{Econometric Specification.} We estimate the rank-rank slopes in \eqref{eq: model with subpops} with a subpopulation $g$ referring to the German state in which the child was born. $Y$ is the child's income, $X$ the parent's income, and $W$ only contains a constant. For each state $g$, we estimate the rank-rank slope $\rho_g$ and the expected rank of a child with parents at the income rank $p=0.25$, i.e., $\theta_{g,0.25}$ defined in \eqref{eq: def theta}. Unlike \cite{Dodin:2024op}, we do not use sample weights. Ranks are defined as in \eqref{eq: def rank} with $\omega=1$. Confidence intervals are based on the plugin estimator of the asymptotic variance in \eqref{eq: AN all coeffs with subpops} (``correct''), based on the homoskedastic (``hom''), or on the Eicker-White (``EW'') variance estimators for the regression in \eqref{eq: model with subpops}.

\paragraph*{Results.} The estimated rank-rank slope for the whole country, i.e. using all observations in our sample, is $0.384$ with a standard error of $0.035$. This estimate indicates relatively low mobility compared to other countries like Australia ($0.215$, \cite{Deutscher:2023oo}), Denmark ($0.203$, \cite{Helso:2021oo}), Italy ($0.220$, \cite{Acciari:2022uu}), or the U.S. ($0.341$, \cite{Chetty:2014tr}).

\begin{figure}[t]
  \centering
  \includegraphics[width=\textwidth]{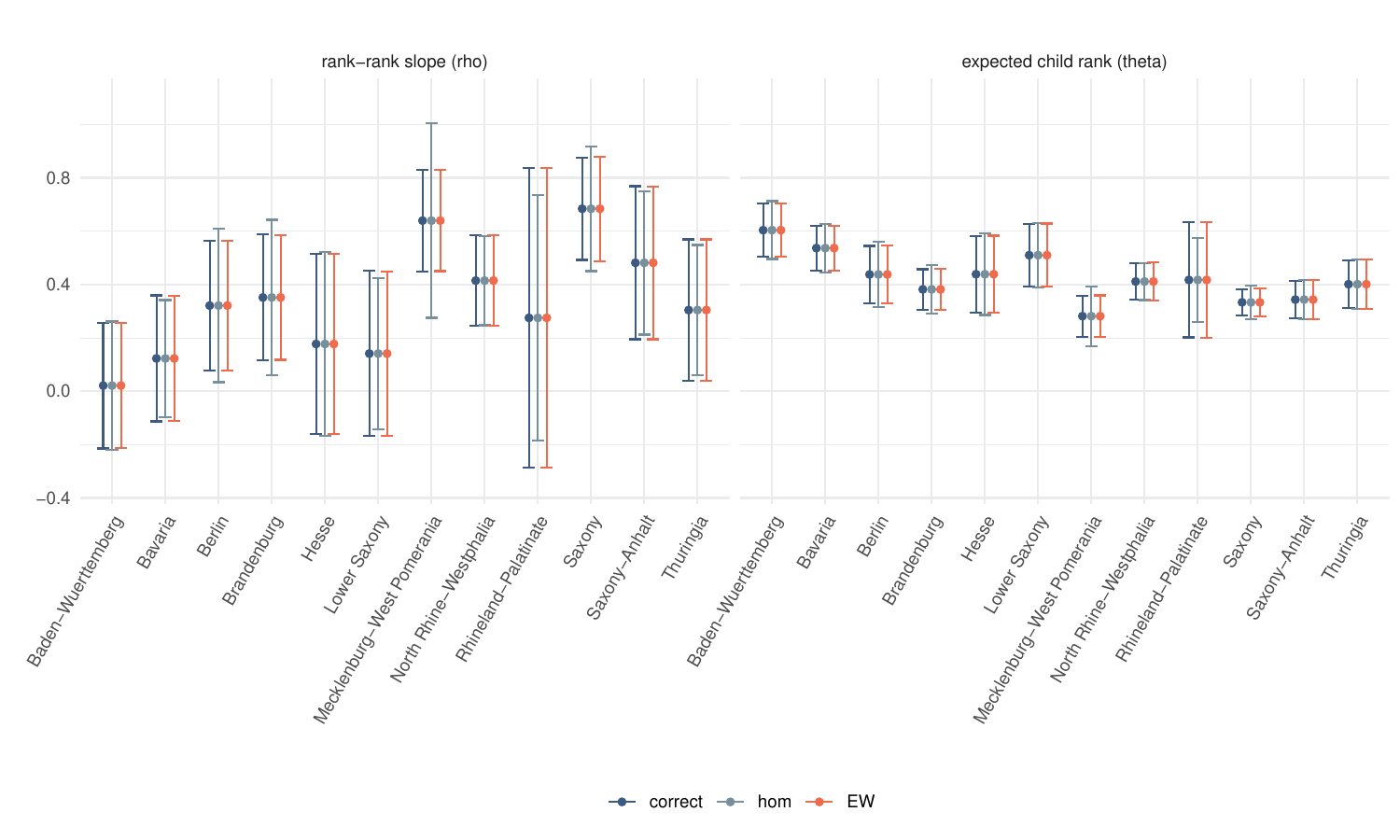}
  \caption{Estimates of the rank-rank slopes $\rho_g$ and expected child ranks $\theta_{g,0.25}$ with 95\% confidence intervals, based on \eqref{eq: model with subpops}. The colors refer to different methods of computing standard errors.}
  \label{fig: SOEP CIs}
\end{figure}

To study the heterogeneity of mobility across states, Figure~\ref{fig: SOEP CIs} shows estimates of the rank-rank slopes ($\rho_g$) and expected child ranks ($\theta_{g,0.25}$) with 95\% confidence intervals for each state. For most states, the rank-rank slope is imprecisely estimated as indicated by wide confidence intervals. The confidence intervals for expected child ranks, on the other hand, are more informative. 

For some states, e.g., Baden-Wuerttemberg, the confidence intervals based on different estimators of the asymptotic variances are similar. However, there are other states, e.g. Mecklenburg-West Pomerania and Rhineland-Palatinate, for which the correct and the hom variance estimates substantially differ. For instance, Mecklenburg-West Pomerania's hom standard error for the rank-rank slope is almost twice as large as the correct one and Rhineland-Palatinate's hom standard error for the expected child rank is 27\% smaller. In this dataset, the EW standard errors turned out to be close to the correct ones for all states and both parameters.

To show that the differences in confidence intervals across methods might matter, we compare the mobility estimates for the German states to those for CZs in the U.S.. To this end, for each German state and for each of the mobility parameters (rank-rank slope and expected child rank), we count how many U.S. CZs have mobility estimates that are larger (smaller) than the upper (lower) bound of the confidence interval of the German state.\footnote{Since the mobility estimates for the U.S. are constructed from the full population of administrative records, we ignore estimation uncertainty in these.} For instance, in the ranking of all 741 CZs, the correct confidence interval for Mecklenburg-West Pomerania's value of $\rho_g$ indicates that its mobility could be placed anywhere from rank 730 to 741, i.e. among the 12 CZs with the lowest mobility.\footnote{The ranking of CZs is in decreasing order of mobility, i.e. the CZ with rank equal to one has the highest mobility. Note that high mobility is indicated by a small (large) value of the rank-rank slope (expected child rank).} In contrast, the hom interval is so wide that it would place Mecklenburg-West Pomerania's mobility among the CZs ranked 130th or worse. Therefore, the correct confidence interval is substantially more informative, allowing us to rule out equal mobility with 729 of the 741 CZs, while the hom confidence interval rules out equal mobility only with 129 of the CZs. 

For instance, the correct confidence interval implies that CZs like Boston, Salt Lake City, San Antonio, Minneapolis, and Washington DC have mobility significantly higher (smaller rank-rank slope) than Mecklenburg-West Pomerania, but the hom confidence interval includes all of these CZs' mobility values.

For some states, the differences in confidence intervals appear small, but the differences occur in regions of mobility values that are attained by many CZs. Therefore, even such apparently small differences can be meaningful. For instance, Saxony's correct confidence interval for the expected child rank excludes the 583 CZs with highest mobility, while the hom confidence interval excludes only 516 of the highest-mobility CZs. So, even though the hom interval is only slightly wider, it includes 67 more CZs than the correct one. 

In summary, the correct confidence interval for mobility parameters in the rank-rank regression may be substantially narrower or wider than those based on other commonly used variance estimators. It therefore may lead to different conclusions in cross-country comparisons of mobility.

\subsection{Intergenerational Education Mobility in India}
\label{sec: Asher et al}

\cite{Asher:2024iu} compute rank-based measures of mobility from data on children's and parents' years of education in India. Because education is observed on discrete points of support, they partially identify rank-rank slopes for a latent, continuously distributed measure of education and then estimate those bounds. We complement their analysis by directly estimating rank-rank slopes from the observed, discrete education data, highlighting the sensitivity of the rank-rank slope to the way in which ties are handled, and by providing valid confidence intervals.

\paragraph*{Data.} We use the 2012 India Human Development Survey (IHDS) dataset provided in the replication package \cite{Asher:2024aa} and focus on the relationship between fathers' and sons' years of education. 

\paragraph*{Econometric Specification.} For each birth cohort of the children, we estimate the rank-rank slopes in \eqref{eq: model with controls} with $W$ containing only a constant. Unlike \cite{Asher:2024iu} we do not use sample weights. Ranks are defined as in \eqref{eq: def rank} with $\omega\in\{0,0.5,1\}$. Confidence intervals for the rank-rank slope are based on the plugin estimator of the asymptotic variance from Section~\ref{sec: const var estim}. Confidence intervals for the rank correlation are computed using the bootstrap.

\paragraph*{Results.} Fathers' and sons' education is observed on seven support points from 0 to 14. Because of this discreteness, we compare rank-rank slope estimates with different ways of handling ties in the ranks. For each birth cohort and each definition of the rank, Figure~\ref{fig: Asher results} reports the three terms in \eqref{eq: OLS estimator}: the estimated rank-rank slope $\hat\rho$, Spearman's rank correlation $\hat\rho_S$, and the ratio of estimated standard deviations $S_Y/S_X$.

\begin{figure}[!t]
  \centering
  \includegraphics[width=\textwidth]{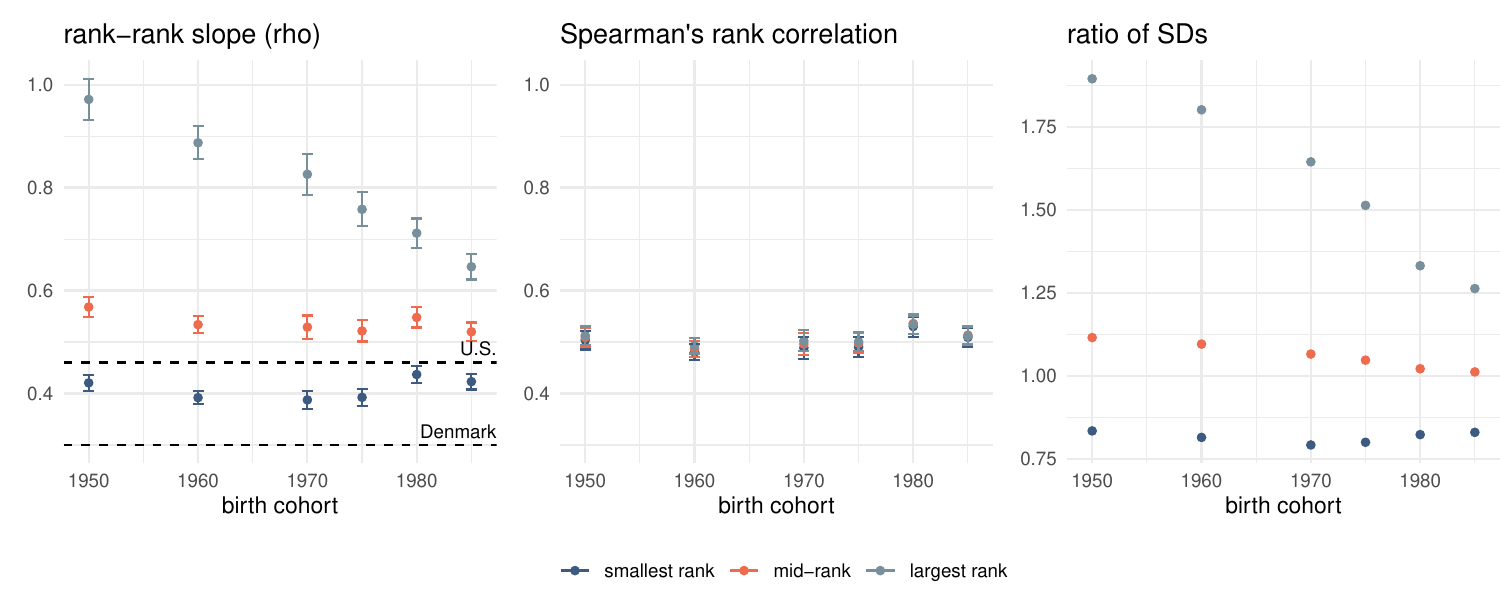}
  \caption{Estimated rank-rank slope $\hat\rho$, Spearman's rank correlation $\hat\rho_S$, and the ratio of standard deviations $S_Y/S_X$ by birth cohort. Bars indicate 95\% confidence intervals. The colors refer to different definitions of the ranks: ``smallest rank'' ($\omega=0$), ``mid-rank'' ($\omega=0.5$), and ``largest rank'' ($\omega=1$). The dashed-lines indicate the values of similar measures of mobility for the U.S. and Denmark as reported in \cite{Asher:2024iu}.}
  \label{fig: Asher results}
\end{figure}

Consider the graph on the left, which shows estimates and confidence intervals for the rank-rank slope. The estimates are very sensitive to the way in which ties are handled, particularly in the earlier birth cohorts. For instance, using the largest rank ($\omega=1$) leads to a rank-rank slope close to one with a relatively wide confidence interval. On the other hand, using the smallest rank ($\omega=0$) leads to a rank-rank slope of $0.42$. In the latest birth cohort (1985), the largest and smallest ranks lead to rank-rank slopes of $0.65$ and $0.42$, still a substantial difference. For the 1950 birth cohort, according to the estimates based on the smallest ranks, mobility in India is higher than in the U.S. (i.e., the rank-rank slope is smaller), while the estimates based on the largest ranks imply that India has extremely low mobility with a rank-rank slope about twice as large as that of the U.S.. Also, based on the largest ranks, mobility in India has increased (rank-rank slopes have decreased) across birth cohorts. According to the estimates based on the smallest ranks, mobility has remained constant at a high level (small rank-rank slope). It is worth emphasizing that, for a given birth cohort, the difference in estimates is due only to the way in which ties in the ranks are handled, i.e. which value of $\omega$ is used. Otherwise, the estimator and the data used are identical.

The graph in the middle shows estimates and confidence intervals for the rank correlation for different ways of handling ties. Interestingly, the estimates and their confidence intervals are almost insensitive to the way in which ties are handled. The estimates are close to $0.5$ for all birth cohorts and all values of $\omega$, and the confidence intervals are all very narrow. To explain why the estimates of the rank-rank slopes and of the rank correlation substantially differ, recall that the rank-rank slope $\hat{\rho}$ is equal to the rank correlation $\hat{\rho}_S$ multiplied by the ratio of standard deviations of the ranks; this is the first equality in \eqref{eq: OLS estimator}, which holds regardless of whether the marginal distributions of $X$ and $Y$ are continuous or not.
If they are both continuous, then the second equality in \eqref{eq: OLS estimator} also holds because the ratio of standard deviations is equal to one. On the other hand, when the marginal distributions are discrete, the ratio can take any value in $(0,\infty)$. While the rank correlation is almost insensitive to how ties are handled and barely varies from birth cohort to birth cohort, the ratio of standard deviations is very sensitive to the way in which ties are handled and varies substantially across birth cohorts. The ratios of standard deviations is displayed in the right-most plot in Figure~\ref{fig: Asher results}. When the smallest ranks ($\omega=0$) are employed, then the ratio of standard deviations is close to $0.8$ for all birth cohorts, while for the largest ranks ($\omega=1$), the ratio starts out at a value close to $1.9$ for the 1950 birth cohort and then decreases towards about $1.25$ for the 1985 birth cohort. Overall, the patterns in the ratios of standard deviations mimic those found in the rank-rank slopes, while the rank correlation remains close to $0.5$ for all birth cohorts and all definitions of the ranks. Therefore, the patterns seen in the rank-rank slopes across birth cohorts and their sensitivity to the way in which ties are handled originate from the variation in the ratios of standard deviations rather than in Spearman's rank correlation.

In summary, in this dataset, it is possible to reach robust conclusions about the value of the rank correlation, namely it being close to $0.5$ for all birth cohorts. On the other hand, the rank-rank slope estimates are so sensitive to the definition of the ranks that they do not provide robust conclusions about whether the trend across birth cohorts is increasing or decreasing, nor whether the rank-rank slope in India is larger or smaller than that of the U.S.. In particular, the trends observed in rank-rank slopes are not due to trends in the rank correlation, but rather due to trends in the marginal distribution of education.

\bibliographystyle{econ-econometrica}
\bibliography{ref}

\clearpage
\newpage

\begin{appendix}
\section{Proof of the Main Result}\label{sec: proof main result}

The proof of Theorem \ref{thm: AN for rank rank reg} relies on a few auxiliary results, which are presented below. As in the main text, throughout this appendix, $\omega\in[0,1]$ is fixed and the same in the definitions of $R_X$, $\hat{R}_X$, $R_Y$, and $\hat{R}_Y$.

\begin{proof}[Proof of Theorem \ref{thm: AN for rank rank reg}]
    By the Frisch-Waugh-Lovell theorem, the estimator $\hat\rho$ in \eqref{eq: joint ols estimator}
    can alternatively be written as
    \begin{equation}\label{eq: hat rho alternative}
        \hat\rho = \frac{\sum_{i=1}^n  R_i^Y( R_i^X - W_i'\hat\gamma)}{\sum_{i=1}^n (R_i^X - W_i'\hat\gamma)^2} = \frac{\sum_{i=1}^n \hat R_Y(Y_i)(\hat R_X(X_i) - W_i'\hat\gamma)}{\sum_{i=1}^n (\hat R_X(X_i) - W_i'\hat\gamma)^2},
    \end{equation}
    where $\hat\gamma = (\sum_{i=1}^n W_i W_i')^{-1}(\sum_{i=1}^n W_i\hat R_X(X_i))$
    is the OLS estimator of a regression of $\hat R_X(X_i)$ on $W_i$. Therefore, using $\sum_{i=1}^n W_i(\hat R_X(X_i) - W_i'\hat\gamma) = 0$ and replacing $\hat R_Y(Y_i)$ in the numerator of \eqref{eq: hat rho alternative} by $\hat R_Y(Y_i) - \rho\hat R_X(X_i) + \rho\hat R_X(X_i)$,
\begin{equation}\label{eq: fwl to be referred in another thm}
    \sqrt n(\hat\rho - \rho) = \frac{\frac{1}{\sqrt n}\sum_{i=1}^n (\hat R_Y(Y_i) - \rho\hat R_X(X_i))(\hat R_X(X_i) - W_i'\hat\gamma)}{\frac{1}{n}\sum_{i=1}^n (\hat R_X(X_i) - W_i'\hat\gamma)^2}.
\end{equation}
    Thus, by Assumption \ref{as: variable nu} and Lemmas \ref{lem: denominator limit} and \ref{lem: three i's bound},
    \begin{equation}\label{eq: rho hat ratio}
        \sqrt n(\hat\rho - \rho) = \frac{\frac{1}{\sqrt n}\sum_{i=1}^n (\hat R_Y(Y_i) - \rho\hat R_X(X_i) - W_i'\beta)(\hat R_X(X_i) - W_i'\gamma)}{\frac{1}{n}\sum_{i=1}^n (\hat R_X(X_i) - W_i'\hat\gamma)^2} + o_P(1).
    \end{equation}
    Consider the numerator. Define $Z_i := (Y_i,X_i,W_i')'$ for all $i=1,\dots,n$ and
    $$f(Z_i,Z_j,Z_k):=\left(I(Y_j, Y_i) - \rho I(X_j, X_i) - W_i'\beta\right)(I(X_k, X_i) - W_i'\gamma)$$ 
    for all $i,j,k=1,\dots,n$. Also, define $h(Z_i,Z_j,Z_k) := 6^{-1}\sum_{i_1,i_2,i_3}f(Z_{i_1},Z_{i_2},Z_{i_3})$, where the sum is over all six permutations of the triplet $(i,j,k)$, 
for all $i,j,k=1,\dots,n$. Note that $h$ is a symmetric function satisfying $E[h(Z_i,Z_j,Z_k)]=0$ whenever $1\leq i<j<k\leq n$. Moreover, let $\Sigma_1$ be the set of triplets $(i,j,k)$ in $\{1,\dots,n\}^3$ such that all three elements are the same, $\Sigma_2$ be the set of triplet $(i,j,k)$ in $\{1,\dots,n\}^3$ such that two out of three elements are the same, and $\Sigma_3$ be the set of triplet $(i,j,k)$ in $\{1,\dots,n\}^3$ such that all three elements are different. It is then easy to check that for all $l=1,2,3$,
$$
\sum_{(i,j,k)\in\Sigma_l} f(Z_i,Z_j,Z_k) = \sum_{(i,j,k)\in\Sigma_l} h(Z_i,Z_j,Z_k),
$$
and so
$$
\sum_{i,j,k=1}^n f(Z_i,Z_j,Z_k) = \sum_{i,j,k=1}^n h(Z_i,Z_j,Z_k).
$$
Therefore,
    \begin{align*}
        &\frac{1}{n}\sum_{i=1}^n (\hat R_Y(Y_i) - \rho\hat R_X(X_i) - W_i'\beta)(\hat R_X(X_i) - W_i'\gamma)\\
        &\qquad=\frac{1}{n^3}\sum_{i,j,k=1}^n \left(I(Y_j, Y_i) - \rho I(X_j, X_i) - W_i'\beta\right)(I(X_k, X_i) - W_i'\gamma) + o_P(n^{-1/2}) \\
        &\qquad = \frac{1}{n^3}\sum_{i,j,k=1}^n f(Z_i,Z_j,Z_k) + o_P(n^{-1/2}) = \frac{1}{n^3}\sum_{i,j,k=1}^n h(Z_i,Z_j,Z_k) + o_P(n^{-1/2})\\
        &\qquad =\begin{pmatrix}n\\ 3\end{pmatrix}^{-1} \sum_{1\leq i<j<k\leq n} h(Z_i,Z_j,Z_k) + o_P(n^{-1/2}),
    \end{align*}
where the last line follows from the Lemma on p. 206 of \cite{Serfling:2002re}, whose application is justified since  $E(h(Z_i,Z_j,Z_k)^2)<\infty$ by Assumption \ref{as: vector w}.
    Furthermore, the results on p. 188 of \cite{Serfling:2002re} imply that the U-statistic can be projected onto the basic observations. To compute the projection, note that, for $z=(y,x,w')'$,
    \begin{align*} 
        E[f(Z_1,Z_2,Z_3)\mid Z_1=z] &= E[f(Z_1,Z_3,Z_2) | Z_1=z]\\
                &= (R_Y(y) -\rho R_X(x)-w'\beta)(R_X(x)-w'\gamma) = h_1(x,w,y),\\
        E[f(Z_2,Z_1,Z_3)\mid Z_1=z] &= E[f(Z_3,Z_1,Z_2)\mid Z_1=z]\\
                &= E\left[(I(y,Y) -\rho I(x,X)-W'\beta)(R_X(X)-W'\gamma)\right] = h_2(x,y),\\
        E[f(Z_2,Z_3,Z_1)\mid Z_1=z] &= E[f(Z_3,Z_2,Z_1)\mid Z_1=z]\\
                &= E\left[(R_Y(Y) -\rho R_X(X)-W'\beta)(I(x,X)-W'\gamma)\right] = h_3(x).
    \end{align*}
    Denoting $\tilde h(z) = E[h(Z_1,Z_2,Z_3) | Z_1 = z]$, we thus have $\tilde h(z)=(h_1(x,w,y)+h_2(x,y)+h_3(x))/3$, and so
    \begin{align*}
        &\begin{pmatrix}n\\ 3\end{pmatrix}^{-1} \sum_{1\leq i<j<k\leq n} h(Z_i,Z_j,Z_k)  = \frac{3}{n}\sum_{i=1}^n \tilde h(Z_i) + o_P(n^{-1/2})\\
            &\quad=\frac{1}{n}\sum_{i=1}^n \Big\{  h_1(X_i,W_i,Y_i) + h_2(X_i,Y_i) + h_3(X_i) \Big\} + o_P(n^{-1/2}).
    \end{align*}
    Therefore, the central limit theorem (e.g., Theorem 9.5.6 in \cite{dudley02}) implies that the numerator of \eqref{eq: rho hat ratio} is asymptotically normal with mean zero and variance $E[(h_1(X,W,Y) + h_2(X,Y) + h_3(X))^2]$. By Lemma~\ref{lem: denominator limit}, the denominator of \eqref{eq: rho hat ratio} converges in probability to $\sigma_{\nu}^2$, so that Slutsky's lemma yields the asserted claim.
\end{proof}

\begin{lemma}\label{lem: gamma asy normality}
    Under Assumptions \ref{as: random sample} and \ref{as: vector w}, we have $\sqrt n(\widehat\gamma - \gamma) = O_P(1)$.
\end{lemma}
\begin{proof}
    By \eqref{eq: first stage}, for $\widehat\gamma$ defined in the beginning of the proof of Theorem \ref{thm: AN for rank rank reg},
    $$
    \sqrt n(\widehat\gamma - \gamma) = \left(\frac{1}{n}\sum_{i=1}^n W_iW_i'\right)^{-1}\left(\frac{1}{\sqrt n}\sum_{i=1}^n W_i(\nu_i + \widehat R_X(X_i) - R_X(X_i))\right).
    $$
    The rest of the proof follows from the law of large numbers, Chebyshev's inequality, and results on p. 183 in \cite{Serfling:2002re}.
\end{proof}

\begin{lemma}\label{lem: glivenko cantelli}
    Under Assumption \ref{as: random sample}, we have $\sup_{x\in\mathbb R}|\hat{R}_X(x) - R_X(x)|=o_P(1)$ and $\sup_{y\in\mathbb R}|\hat R_Y(y) - R_Y(y)|=o_P(1)$.
\end{lemma}
\begin{proof}
    By the Glivenko-Cantelli theorem (e.g., Theorem 1.3 in \cite{dudley2014}), $\sup_{x\in\mathbb R}|\hat F_X(x) - F_X(x)|=o_P(1)$. Also, by the Glivenko-Cantelli theorem applied to $-X_i$'s instead of $X_i$'s, $\sup_{x\in\mathbb R}|\hat F_X^-(x) - F_X^-(x)|=o_P(1)$. The first claim thus follows. The second claim follows from the same argument.
\end{proof}

\begin{lemma}\label{lem: denominator limit}
    Under Assumptions \ref{as: random sample} and \ref{as: vector w}, we have 
    $$\frac{1}{n}\sum_{i=1}^n(\widehat R_X(X_i)-W_i'\widehat\gamma)^2\to_P E[(R_X(X)-W'\gamma)^2] = \sigma_{\nu}^2.$$
\end{lemma}

\begin{proof}
The proof follows from the law of large numbers, the elementary identity $a^2 - b^2 = (a-b)^2 +2(a-b)b$, and Lemmas \ref{lem: gamma asy normality} and \ref{lem: glivenko cantelli}.
\end{proof}

\begin{lemma}\label{lem: three i's bound}
    Under Assumptions \ref{as: random sample} and \ref{as: vector w}, we have
    \begin{align*}
    & \frac{1}{\sqrt n}\sum_{i=1}^n(\hat R_Y(Y_i) - \rho\hat R_X(X_i))(\hat R_X(X_i) - W_i'\hat\gamma) \\
    &\qquad = \frac{1}{\sqrt n}\sum_{i=1}^n(\hat R_Y(Y_i) - \rho\hat R_X(X_i) - W_i'\beta)(\hat R_X(X_i) - W_i'\gamma) + o_P(1).
    \end{align*}
\end{lemma}

\begin{proof}
The proof follows from the identity $\sum_{i=1}^n W_i(\hat R_X(X_i) - W_i'\hat\gamma) = 0$, Lemmas \ref{lem: gamma asy normality} and \ref{lem: glivenko cantelli}, and Chebyshev's inequality.
\end{proof}

\section{Asymptotic Theory for All Regression Coefficients}\label{sec: all coefficients}

\subsection{Rank-Rank Regressions}
\label{app: AN all coefficients rank-rank}

In this section, we present the asymptotic theory for all coefficients in the rank-rank regression model studied in Section~\ref{sec: inf rank rank reg}:
\begin{equation}\label{eq: model with controls app}
    R_Y(Y) = \rho R_X(X) + W'\beta + \varepsilon,\qquad E\left[\varepsilon\begin{pmatrix}
    R_{X}(X)\\ W
    \end{pmatrix}\right]=0,
\end{equation}
where we had already introduced the equation projecting $R_X(X)$ onto the other regressors in $W$:
\begin{equation}\label{eq: first stage app 1}
    R_X(X) = W'\gamma + \nu,\quad E[\nu W] = 0.
\end{equation}
To study the asymptotic behavior of the coefficients $\beta$, we now also introduce some additional projections. Let $W_l$ denote an element of $W := (W_1,\ldots,W_p)'$ and by $W_{-l}$ the vector of all elements of $W$ except the $l$-th. We now introduce the projection of $W_l$ onto $R_X(X)$ and the remaining regressors $W_{-l}$: for any, $l=1,\ldots, p$,
\begin{equation}\label{eq: first stage app 2}
    W_l = \tau_l R_X(X) + W_{-l}'\delta_l + \upsilon_l,\quad E\left[\upsilon_l \begin{pmatrix} R_{X}(X)\\ W_{-l} \end{pmatrix}\right] = 0,
\end{equation}
where $\tau_1,\ldots,\tau_p$ are scalar constants and $\delta_1,\ldots,\delta_{p}$ are $(p-1)$-dimensional vectors of constants.

\begin{assumption}\label{as: variable upsilon}
    For $l=1,\ldots, p$, the random variable $\upsilon_l$ is such that $E[\upsilon_l^2]>0$.
\end{assumption}
Consider the OLS estimator in \eqref{eq: joint ols estimator}:
\begin{equation}\label{eq: joint ols estimator app}
\begin{pmatrix}
\hat{\rho}\\
\hat{\beta}
\end{pmatrix}=\left(\sum_{i=1}^{n}\begin{pmatrix}
R_i^X\\
W_{i}
\end{pmatrix}\begin{pmatrix}
R_i^X & W_{i}'\end{pmatrix}\right)^{-1}\sum_{i=1}^{n}\begin{pmatrix}
R_i^X\\
W_{i}
\end{pmatrix}R_i^Y.
\end{equation}

Let $\sigma_{\nu}^2$, $h_{0,1} := h_1$, $h_{0,2}:= h_2$, and $h_{0,3}:=h_3$ be as defined in Theorem~\ref{thm: AN for rank rank reg}. For $l=1,\ldots, p$, further define $\sigma_{\upsilon_l}^2 := E[\upsilon_l^2]$, 
\begin{align*}
    \phi_0(x,w,y) &:= \frac{1}{\sigma_{\nu}^2}\left[h_{0,1}(x,w,y)+h_{0,2}(x,y)+h_{0,3}(x) \right]\\
    \phi_l(x,y,w) &:= \frac{1}{\sigma_{\upsilon_l}^2}\left[h_{l,1}(x,w,y)+h_{l,2}(x,y)+h_{l,3}(x) \right]
\end{align*}
and
\begin{align*}
    h_{l,1}(x,w,y) &:= (R_Y(y) - \rho R_X(x) - w'\beta)(w_l - \tau_l R_X(x) - w_{-l}'\delta_l),\\
    h_{l,2}(x,y) & := E[(I(y,Y) - \rho I(x,X) - W'\beta)(W_l - \tau_l R_X(X) - W_{-l}'\delta_l)], \\
    h_{l,3}(x) & := E[(R_Y(Y) - \rho R_X(X) - W'\beta)(W_l - \tau_l I(x,X) -  W_{-l}'\delta_l)].
\end{align*}
Finally, let $\psi_i := (\phi_0(X_i,W_i,Y_i), \phi_1(X_i,W_i,Y_i), \ldots, \phi_p(X_i,W_i,Y_i))'$.

With the additional assumption and notation, one can derive the following joint asymptotic normality result:

\begin{theorem}\label{thm: AN for rank rank reg all coeffs}
    Suppose that \eqref{eq: model with controls app}--\eqref{eq: joint ols estimator app} hold and that Assumptions~\ref{as: random sample}--\ref{as: variable nu}, and \ref{as: variable upsilon} are satisfied. Then
    $$
    \sqrt n  \begin{pmatrix}\hat \rho - \rho \\ \hat\beta -\beta \end{pmatrix}  = \frac{1}{\sqrt n}\sum_{i=1}^n \psi_i + o_P(1)\to_D N(0,\Sigma),
    $$
    where $\Sigma := E[\psi_i \psi_i']$ and $\psi_i$'s are as defined above in this section.
\end{theorem}

Theorem~\ref{thm: AN for rank rank reg} has already shown the asymptotic variance of $\hat\rho$. Theorem \ref{thm: AN for rank rank reg all coeffs} adds to that result by providing the asymptotic variance for individual components of $\hat\beta$: for all $l=1,\dots,p$, we have
$$\sqrt{n}(\hat\beta_l-\beta_l) \to_D N(0,\sigma_{\beta_l}^2), $$
where
$$\sigma_{\beta_l}^2 := \frac{1}{\sigma_{\upsilon_l}^4}E\left[\left(h_{l,1}(X,W,Y)+h_{l,2}(X,Y)+h_{l,3}(X)\right)^2 \right]. $$

\subsection{Rank-Rank Regressions With Subpopulations}
\label{app: AN all coefficients rank-rank with subpops}

In this section, we present the asymptotic theory for all coefficients in the rank-rank regression model with subpopulations studied in Section~\ref{sec: rank rank reg with subpops}:
\begin{equation}\label{eq: model with subpops app}
    R_Y(Y) = \sum_{g=1}^{n_G} \ind\{G=g\}\left(\rho_g R_X(X) + W'\beta_g\right) + \varepsilon,\quad E\left[\left. \varepsilon\begin{pmatrix} R_{X}(X)\\ W\end{pmatrix}\right| G\right]=0\;\text{a.s.},
\end{equation}
where we had already introduced the equation projecting $R_X(X)$ onto the other regressors in $W$:
\begin{equation}\label{eq: first stage with subpops app 1}
    R_X(X) = \sum_{g=1}^{n_G} \ind\{G=g\} W'\gamma_g + \nu,\quad E[\left.\nu W\right| G] = 0\;\text{a.s.}.
\end{equation}
To study the asymptotic behavior of the coefficients $\beta_g$, we now also introduce some additional projections. Let $W_l$ denote an element of $W := (W_1,\ldots,W_p)'$ and by $W_{-l}$ the vector of all elements of $W$ except the $l$-th. We now introduce the projection of $W_l$ onto $R_X(X)$ and the remaining regressors $W_{-l}$: for any, $l=1,\ldots, p$,
\begin{equation}\label{eq: first stage with subpops app 2}
    W_l = \sum_{g=1}^{n_G} \ind\{G=g\}\left(\tau_{l,g} R_X(X) + W_{-l}'\delta_{l,g}\right) + \upsilon_l,\  E\left[\left.\upsilon_l \begin{pmatrix} R_{X}(X)\\ W_{-l} \end{pmatrix}\right| G\right] = 0\;\text{a.s.},
\end{equation}
where $\tau_{l,g}$ are scalar constants and $\delta_{l,g}$ are $(p-1)$-dimensional vectors of constants.

\begin{assumption}\label{as: variable upsilon with subpops}
    For $l=1,\ldots, p$ and $g=1,\ldots,n_G$, the random variable $\upsilon_l$ is such that $E[\ind\{G=g\}\upsilon_l^2]>0$.
\end{assumption}
Consider the OLS estimator in \eqref{eq: joint ols estimator with subpops}:
\begin{equation}\label{eq: joint ols estimator with subpops app}
\begin{pmatrix}
\hat{\rho}_g\\
\hat{\beta}_g
\end{pmatrix}=\left(\sum_{i=1}^{n}\ind\{G_i=c\}\begin{pmatrix}
R_i^X\\
W_{i}
\end{pmatrix}\begin{pmatrix}
R_i^X & W_{i}'\end{pmatrix}\right)^{-1}\sum_{i=1}^{n}\ind\{G_i=c\}\begin{pmatrix}
R_i^X\\
W_{i}
\end{pmatrix}R_i^Y.
\end{equation}
Let $v_0 := \nu$ and for $l=0,\ldots, p$ and $g=1,\ldots,n_G$, define $\sigma_{\upsilon_l,g}^2 := E[\ind\{G=g\}\upsilon_l^2]$. Also, for $l=1,\ldots, p$ and $g=1,\ldots,n_G$, define
\begin{align*}
\xi_{0,1,g}(x,w) & := R_X(x) - w'\gamma_g, \quad \xi_{l,1,g}(x,w) := w_l - \tau_{l,g}R_X(x) - w_{-l}'\delta_{l,g},\\ 
\xi_{0,3,g}(x,X,W)& := I(x,X)-W'\gamma_g, \quad \xi_{l,3,g}(x,X,W) := W_l - \tau_{l,g}I(x,X) - W_{-l}'\delta_{l,g}.
\end{align*}
Moreover, $l=0,\ldots, p$ and $g=1,\ldots,n_G$, define
\begin{align*}
    h_{l,1,g}(\bar{g},x,w,y) := \ind\{\bar{g}=g\}(R_Y(y) - \rho_g R_X(x) - w'\beta_g)\xi_{l,1,g}(x,w),\\
    h_{l,2,g}(x,y)  := E[\ind\{G=g\}(I(y,Y) - \rho_g I(x,X) - W'\beta_g)\xi_{l,1,g}(X,W)], \\
    h_{l,3,g}(x)  := E[\ind\{G=g\}(R_Y(Y) - \rho_g R_X(X) - W'\beta_g)\xi_{l,3,g}(x,X,W)],\\
    \phi_{l,g}(\bar{g},x,y,w) := \frac{1}{\sigma_{\upsilon_{l},g}^2}\left[h_{l,1,g}(\bar{g},x,w,y)+h_{l,2,g}(x,y)+h_{l,3,g}(x) \right].
\end{align*}
Finally, let $\psi_{i,g} := (\phi_{0,g}(G_i,X_i,W_i,Y_i), \phi_{1,g}(G_i,X_i,W_i,Y_i), \ldots, \phi_{p,g}(G_i,X_i,W_i,Y_i))'$ and 
\begin{align*}
    \psi_i &:= (\phi_{0,1}(G_i,X_i,W_i,Y_i), \ldots, \phi_{0,n_G}(G_i,X_i,W_i,Y_i), \\
    &\qquad \phi_{1,1}(G_i,X_i,W_i,Y_i), \ldots, \phi_{1,n_G}(G_i,X_i,W_i,Y_i), \ldots)'.
\end{align*}
With the additional assumption and notation, one can derive the following joint asymptotic normality result:

\begin{theorem}\label{thm: AN for rank rank reg all coeffs subpops}
    Suppose \eqref{eq: model with subpops app}--\eqref{eq: joint ols estimator with subpops app} hold and that Assumptions~\ref{as: random sample with subpops}--\ref{as: large subpops}, and \ref{as: variable upsilon with subpops} are satisfied. Then:
    $$
    \sqrt n  \begin{pmatrix}\hat \rho - \rho \\ \hat\beta -\beta \end{pmatrix}  = \frac{1}{\sqrt n}\sum_{i=1}^n \psi_i + o_P(1)\to_D N(0,\Sigma),
    $$
    where $\Sigma := E[\psi_i \psi_i']$ and $\psi_i$ as defined in this section.
\end{theorem}

Theorem~\ref{thm: AN for rank rank reg with subpops} already showed the asymptotic normality of $\hat\rho_g$. Theorem \ref{thm: AN for rank rank reg all coeffs subpops} in turn shows the asymptotic normality of $\hat\beta_g$ as well and provides explicit formulas for the asymptotic variance of both $\hat\rho_g$ and $\hat\beta_g$.

\subsection{Regression of a General Outcome on a Rank}
\label{app: AN all coefficients outcome-rank}

In this section, we present the asymptotic theory for all coefficients in the rank-rank regression model studied in Section~\ref{sec: outcome on rank}:
\begin{equation}\label{eq: model of outcome on rank app}
    Y = \rho R_X(X) + W'\beta + \varepsilon,\qquad E\left[\varepsilon\begin{pmatrix}
    R_{X}(X)\\ W
    \end{pmatrix}\right]=0,
\end{equation}
with the projection equations \eqref{eq: first stage app 1} and \eqref{eq: first stage app 2}. Consider the OLS estimator in \eqref{eq: joint ols estimator outcome on rank}:
\begin{equation}\label{eq: joint ols estimator outcome on rank app}
    \begin{pmatrix}
    \hat{\rho}\\
    \hat{\beta}
    \end{pmatrix}=\left(\sum_{i=1}^{n}\begin{pmatrix}
    R_i^X\\
    W_{i}
    \end{pmatrix}\begin{pmatrix}
    R_i^X & W_{i}'\end{pmatrix}\right)^{-1}\sum_{i=1}^{n}\begin{pmatrix}
    R_i^X\\
    W_{i}
    \end{pmatrix}Y_i.
\end{equation}
Let $v_0 := \nu$ and for $l=0,\ldots, p$, define $\sigma_{\upsilon_l}^2 := E[\upsilon_l^2]$. Also, for $l=1,\ldots, p$, define
\begin{align*}
\xi_{0,1}(x,w) & := R_X(x) - w'\gamma, \quad \xi_{l,1}(x,w) := w_l - \tau_{l}R_X(x) - w_{-l}'\delta_{l},\\ 
\xi_{0,3}(x,X,W)& := I(x,X)-W'\gamma, \quad \xi_{l,3}(x,X,W) := W_l - \tau_{l}I(x,X) - W_{-l}'\delta_{l}.
\end{align*}
Moreover, $l=0,\ldots, p$, define
\begin{align*}
    h_{l,1}(x,w,y) &:= (y - \rho R_X(x) - w'\beta)\xi_{l,1}(x,w),\\
    h_{l,2}(x) & := E[(Y - \rho I(x,X) - W'\beta)\xi_{l,1}(X,W)], \\
    h_{l,3}(x) & := E[(Y - \rho R_X(X) - W'\beta)\xi_{l,3}(x,X,W)],\\
    \phi_l(x,y,w) &:= \frac{1}{\sigma_{\upsilon_l}^2}\left[h_{l,1}(x,w,y)+h_{l,2}(x)+h_{l,3}(x) \right]
\end{align*}
Finally, let $\psi_i := (\phi_{0}(X_i,W_i,Y_i), \phi_{1}(X_i,W_i,Y_i), \ldots, \phi_{p}(X_i,W_i,Y_i))'$. With this additional notation, one can derive the following joint asymptotic normality result:

\begin{theorem}\label{thm: AN for rank rank reg all coeffs gen outcome on rank}
    Suppose \eqref{eq: first stage app 1}, \eqref{eq: first stage app 2}, \eqref{eq: model of outcome on rank app}, and \eqref{eq: joint ols estimator outcome on rank app} hold and that Assumptions~\ref{as: random sample}--\ref{as: variable nu}, and \ref{as: variable upsilon} are satisfied. Then:
    $$
    \sqrt n  \begin{pmatrix}\hat \rho - \rho \\ \hat\beta -\beta\end{pmatrix}  = \frac{1}{\sqrt n}\sum_{i=1}^n \psi_i + o_P(1)\to_D N(0,\Sigma),
    $$
    where $\Sigma := E[\psi_i \psi_i']$ and $\psi_i$ as defined in this section.
\end{theorem}

Theorem~\ref{thm: AN for outcome on rank reg} already showed the asymptotic normality of $\hat\rho$. Theorem \ref{thm: AN for rank rank reg all coeffs gen outcome on rank} in turn shows the asymptotic normality of $\hat\beta$ as well and provides explicit formulas for the asymptotic variance of both $\hat\rho$ and $\hat\beta$.

\section{Additional Proofs}
\label{app: add proofs}

\subsection{Proofs for Section \ref{sec: motivation} }

Define the sample counterpart to $M_{kl}$:
$$\hat{M}_{kl} :=\frac{1}{n} \sum_{i=1}^n \left(\hat{F}_X(X_i)-\frac{1}{n}\sum_{j=1}^n\hat{F}_X(X_j)\right)^k\left(\hat{F}_Y(Y_i)-\frac{1}{n}\sum_{j=1}^n\hat{F}_Y(Y_j)\right)^l. $$
Before proving Lemma \ref{lem: plims}, we state the following auxiliary lemma.

\begin{lemma}\label{lem: Mhat consistency}
    Let $(X_i,Y_i)$, $i=1,\ldots,n$, be an i.i.d. sample from a distribution $F$ with continuous marginals. Then $\hat{M}_{kl} = M_{kl}+o_P(1)$ for $(k,l)\in \{(0,2),(2,0),(2,2),(3,1),(4,0)\}$.
\end{lemma}

\begin{proof}
By the Glivenko-Cantelli theorem (e.g., Theorem 1.3 in \cite{dudley2014}),
$$
\max\left\{\sup_{x\in\mathbb{R}}|\hat{F}_X(x)-F_X(x)|,\sup_{y\in\mathbb{R}}|\hat{F}_Y(y)-F_Y(y)|\right\} = o_P(1).
$$
The proof of the lemma follows from standard arguments using this convergence result, the Cauchy-Schwarz and triangle inequalities, and some elementary identities such as $a^2-b^2=(a-b)^2+2(a-b)b$ and $a^3-b^3 = (a-b)^3 + 3ab(a-b)$.
\end{proof}

\begin{proof}[Proof of Lemma~\ref{lem: plims}]
Consider first the estimator for homoskedastic regression errors $\sigma_{hom}^2$. We have
    \begin{align*}
      \frac{1}{n}\sum_{i=1}^n \hat{\varepsilon}_i^2 &= \hat{M}_{02} - \hat{\rho}^2 \hat{M}_{20}= M_{02}-\rho^2M_{20} + o_P(1) = M_{02}(1-\rho^2) + o_P(1)
    \end{align*}
    where the first equality follows from simple algebra, the second from Lemma~\ref{lem: Mhat consistency}, and the third from noting that $M_{20} = M_{02}$.
    Similarly,
    $n^{-1}\sum_{i=1}^n (R_i^X-\bar{R}^X)^2 = M_{02} + o_P(1)$
    and so \eqref{eq: hom plim} follows.
    Next, consider the Eicker-White estimator $\sigma_{EW}^2$. Similarly as above, we have
    \begin{align*}
        &\frac{1}{n}\sum_{i=1}^n \hat{\varepsilon}_i^2(R_i^X-\bar{R}^X)^2 = \hat{M}_{22}-2\hat{\rho}\hat{M}_{31}+\hat{\rho}^2\hat{M}_{40} = M_{22} -2\rho M_{31} + \rho^2 M_{40} + o_P(1).
    \end{align*}
In addition,
    $(n^{-1}\sum_{i=1}^n (R_i^X-\bar{R}^X)^2)^2 = M_{20}^2 +o_P(1), $
    and so \eqref{eq: EW plim} follows by noting that $M_{20}=1/12$ and $M_{40}=1/80$.
\end{proof}

\begin{proof}[Proof of Lemma \ref{lem: large ratio of variances}]
Let $a\in(0,1)$ be a constant and let $F$ be the cdf of a pair $(X,Y)$ of $U[0,1]$ random variables such that $Y = (a-X)\ind\{0\leq X\leq a\} + X\ind\{a<X\leq 1\}$.
Then
\begin{align*}
E[\ind\{X\leq x\}Y] & = (ax-x^2/2)\ind\{0\leq x\leq a\} + 2^{-1}x^2\ind\{a<x\leq 1\},\\
E[\ind\{Y\leq y\}X] & = (ay-y^2/2)\ind\{0\leq y\leq a\} + 2^{-1}y^2\ind\{a<y\leq 1\}.
\end{align*}
Hence, given that $X>a$ if and only if $Y>a$, the function $h$ in \eqref{eq: function h true} satisfies
\begin{align*}
E[h(X,Y)] & = E[-(a^2/2)\ind\{X\leq a\}] = -a^3/2, \\
E[h(X,Y)^2] & = E[(a^2/2)^2\ind\{X\leq a\}] = a^5/4.
\end{align*}
Therefore, given that $\sigma^2 = 144Var(h(X,Y))$ by the discussion in a remark after Theorem \ref{thm: AN for rank rank reg}, it follows that
$
\sigma^2 = 144(E[h(X,Y)^2] - (E[h(X,Y)])^2) = 36(a^5 - a^6).
$
On the other hand, by Lemma \ref{lem: plims},
\begin{align*}
\sigma^2_{hom} & = 1 - \rho^2 = 1 - \left(12E[X Y] - 3\right)^2 = -8 - 144(E[XY])^2 + 72E[XY] \\
& = - 8 - 144\left(\frac{1}{3}-\frac{a^3}{6}\right)^2 + 72\left(\frac{1}{3}-\frac{a^3}{6}\right) = 4(a^3 - a^6).
\end{align*}
In addition, tedious algebra shows that
$$
M_{22} = \frac{1}{80}-\frac{a^3}{6} + \frac{a^4}{3} - \frac{a^5}{6}, \quad M_{31} = \frac{1}{80} - \frac{a^3}{8} + \frac{a^4}{4} - \frac{3a^5}{20}, \quad\rho = 1 - 2a^3.
$$
Thus, by Lemma \ref{lem: plims},
\begin{align*}
\sigma_{EW}^2 &= 144\bigg( \left(\frac{1}{80}-\frac{a^3}{6} + \frac{a^4}{3} - \frac{a^5}{6}\right) \\
& \qquad - \left( \frac{1}{40} - \frac{3a^3}{10} + \frac{a^4}{2} - \frac{3a^5}{10} + \frac{a^6}{2} - a^7 + \frac{3a^8}{5} \right) + \left(\frac{1}{80} - \frac{a^3}{20} + \frac{a^6}{20}\right) \bigg) \\
& = 144\left(\frac{a^3}{12} - \frac{a^4}{6} + \frac{2a^5}{15} - \frac{9a^6}{20} + a^7 -\frac{3a^8}{5}\right).
\end{align*}
Combining these bounds, we obtain that $\sigma^2_{hom} / \sigma^2 \to \infty$ and $\sigma_{EW}^2/\sigma^2\to\infty$ as $a\to 0$, yielding the asserted claim.
\end{proof}

\begin{proof}[Proof of Lemma~\ref{lem: small ratio of variances}]
Let $F$ be any distribution with continuous marginals and let $(X,Y)$ be a pair of random variables with distribution $F$. Letting $F_X$ and $F_Y$ denote the corresponding marginal distributions, it follows that $U:=F_X(X)$ and $V:=F_Y(Y)$ are both $U[0,1]$ random variables. Also, $F(X,Y) = C(U,V)$, where $C$ is the copula of $F$. In addition, $\rho = Corr(F_X(X),F_Y(Y)) = 12E[UV]-3$. 

Consider first the case $\rho\geq 0$. By Lemma \ref{lem: plims},
\begin{align*}
\sigma_{hom}^2 
& = 1 - \rho^2 = (1-\rho)(1+\rho) \geq 1 - \rho = 4 - 12 E[UV]\\
& = 4 - E[12UV - 6U^2 - 6V^2 + 6U^2 + 6V^2] \\
& = 4 + 6E[(U-V)^2] - 6E[U^2] - 6E[V^2] = 6E[(U-V)^2].
\end{align*}
On the other hand, by \eqref{eq: avar spearman},
\begin{align}
& \sigma^2/9 
 = Var\left( (2U-1)(2V-1) + 4\int_0^1 (C(U,v) - Uv)dv + 4\int_0^1 (C(u,V) - uV)du \right) \nonumber \\
& \ \leq 2E\left[ \left((2U-1)(2V-1) + 4\int_0^1 (U\wedge v - Uv)dv + 4\int_0^1 (u\wedge V - uV)du - 1\right)^2 \right] \nonumber \\
& \ +2 E\left[\left( 4\int_0^1 (C(U,v) - U\wedge v)dv + 4\int_0^1 (C(u,V) - u\wedge V)du  \right)^2\right]. \label{eq: lem 2 decom 2}
\end{align}
Here, observe that
$$
\int_0^1 (U\wedge v - Uv)dv = \frac{U - U^2}{2}\quad\text{and}\quad \int_0^1 (u\wedge V - uV)du = \frac{V-V^2}{2}.
$$
Thus, the first term on the right-hand side of \eqref{eq: lem 2 decom 2} is equal to
$$
2E\left[( 4UV - 2U^2 - 2 V^2 )^2\right] = 8E\left[(U-V)^4\right] \leq 8E\left[(U-V)^2\right].
$$
To bound the second term on the right-hand side of \eqref{eq: lem 2 decom 2}, we claim that
\begin{equation}\label{eq: lem 2 claim 1}
|C(u,v) - u\wedge v| \leq P(|U-V| \geq |u-v|).
\end{equation}
Indeed, if $u\leq v$, then
\begin{align*}
C(u,v) - u\wedge v 
& = E[\ind\{U\leq u\}\ind\{V\leq v\}] - E[\ind\{U\leq u\}\ind\{U\leq v\}] \\
& = E[\ind\{U\leq u\}(\ind\{V\leq v\} - \ind\{U\leq v\})]  = - E[\ind\{U\leq u\}\ind\{V> v\}],
\end{align*}
and so \eqref{eq: lem 2 claim 1} follows. On the other hand, if $u>v$, then
\begin{align*}
C(u,v) - u\wedge v 
& = E[\ind\{U\leq u\}\ind\{V\leq v\}] - E[\ind\{V\leq u\}\ind\{V\leq v\}] \\
& = E[\ind\{V\leq v\}(\ind\{U\leq u\} - \ind\{V\leq u\})]  = - E[\ind\{V\leq v\}\ind\{U> u\}],
\end{align*}
and so \eqref{eq: lem 2 claim 1} follows as well, yielding the claim. Thus,
\begin{align*}
\int_0^1 |C(u,v) - u\wedge v|dv
& \leq \int_0^1 P(|U-V|\geq |u-v|)dv \\
& \leq 2\int_0^1 P(|U-V|\geq t)dt \leq 2E[|U-V|]
\end{align*}
and, similarly,
$$
\int_0^1 |C(u,v) - u\wedge v|du \leq 2E[|U-V|].
$$
Hence, the expression in \eqref{eq: lem 2 decom 2} is bounded from above by
$$
2\times 16^2 (E[|U-V|])^2 \leq 512 E\left[|U-V|^2\right].
$$
Therefore, $\sigma^2\leq 9\times 520E[|U-V|^2]$, and so $\sigma^2_{hom}/\sigma^2$ is bounded below from zero. 

In turn, the case $\rho<0$ can be treated similarly. The asserted claim follows.
\end{proof}

\subsection{Proofs for Section~\ref{sec: inf rank rank reg}}

\begin{lemma}\label{lem: glivenko cantelli 2}
Under Assumptions \ref{as: random sample} and \ref{as: vector w}, we have
\begin{equation}\label{eq: gc 2-1}
\sup_{y\in\mathbb R}\left| \frac{1}{n}\sum_{i=1}^n I(y,Y_i)(R_X(X_i) - W_i'\gamma) - E[I(y,Y)(R_X(X) - W'\gamma)] \right| = o_P(1),
\end{equation}
\begin{equation}\label{eq: gc 2-2}
\sup_{x\in\mathbb R}\left| \frac{1}{n}\sum_{i=1}^n I(x,X_i)(R_X(X_i) - W_i'\gamma) - E[I(x,X)(R_X(X) - W'\gamma)] \right| = o_P(1),
\end{equation}
\begin{align}
&\sup_{x\in\mathbb R}\bigg| \frac{1}{n}\sum_{i=1}^n(R_Y(Y_i) - \rho R_X(X_i) - W_i'\beta)I(x,X_i) \nonumber\\
&\qquad\qquad\qquad\qquad\qquad - E[(R_Y(Y) - \rho R_X(X) - W'\beta)I(x,X)]\bigg| = o_P(1). \label{eq: gc 2-3}
\end{align}
\end{lemma}
\begin{proof}
For any random variable $A$, let $A^{+}$ and $A^{-}$ denote $A\ind\{A\geq 0\}$ and $-A\ind\{A<0\}$, respectively. Then
$$
\sup_{y\in\mathbb R}\left| \frac{1}{n}\sum_{i=1}^n I(y,Y_i)(R_X(X_i) - W_i'\gamma)^+ - E[I(y,Y)(R_X(X) - W'\gamma)^+] \right| = o_P(1)
$$
and
$$
\sup_{y\in\mathbb R}\left| \frac{1}{n}\sum_{i=1}^n I(y,Y_i)(R_X(X_i) - W_i'\gamma)^- - E[I(y,Y)(R_X(X) - W'\gamma)^-] \right| = o_P(1)
$$
by the argument parallel to that used in the proof of the Glivenko-Cantelli theorem (e.g., Theorem 1.3 in \cite{dudley2014}). Combining these bounds gives \eqref{eq: gc 2-1}. In turn, \eqref{eq: gc 2-2} and \eqref{eq: gc 2-3} follow from the same argument.
\end{proof}

\begin{lemma}\label{lem: simple square approximation}
For any vectors $A = (A_1,\dots,A_n)'$ and $\widehat A = (\widehat A_1,\dots,\widehat A_n)'$, we have
$$
\left|\frac{1}{n}\sum_{i=1}^n (\widehat A_i^2 - A_i^2)\right| \leq \frac {1}{n}\sum_{i=1}^n (\widehat A_i - A_i)^2 + 2\sqrt{\frac {1}{n}\sum_{i=1}^n (\widehat A_i - A_i)^2}\sqrt{\frac {1}{n}\sum_{i=1}^n A_i^2}.
$$
\end{lemma}
\begin{proof}
The proof follows from the Cauchy-Schwarz and triangle inequalities.
\end{proof}

\begin{lemma}\label{lem: simple square approximation 2}
For any vectors $B = (B_1,\dots,B_n)'$, $\widehat B = (\widehat B_1,\dots,\widehat B_n)'$, $C = (C_1,\dots,C_n)'$, and $\widehat C = (\widehat C_1,\dots,\widehat C_n)'$, we have
\begin{align*}
&\sqrt{\frac{1}{n}\sum_{i=1}^n (\widehat B_i \widehat C_i - B_iC_i)^2}
\leq  \left(\frac{1}{n}\sum_{i=1}^n(\widehat B_i - B_i)^4\right)^{1/4} \left(\frac{1}{n}\sum_{i=1}^n(\widehat C_i - C_i)^4\right)^{1/4} \\
&\quad  +  \left(\frac{1}{n}\sum_{i=1}^n C_i^4\right)^{1/4}\left(\frac{1}{n}\sum_{i=1}^n(\widehat B_i - B_i)^4\right)^{1/4}  + \left(\frac{1}{n}\sum_{i=1}^n B_i^4\right)^{1/4} \left(\frac{1}{n}\sum_{i=1}^n(\widehat C_i - C_i)^4\right)^{1/4}.
\end{align*}
\end{lemma}
\begin{proof}
The proof follows from the Cauchy-Schwarz and triangle inequalities.
\end{proof}

\begin{proof}[Proof of Lemma \ref{lem: consistent variance estimation}]
First, we prove that
\begin{equation}\label{eq: cons variance est 1}
\left| \frac{1}{n}\sum_{i=1}^n(H_{1i} + H_{2i} + H_{3i})^2 - \frac{1}{n}\sum_{i=1}^n (h_1(X_i,W_i,Y_i) + h_2(X_i,Y_i) + h_3(X_i))^2 \right| = o_P(1).
\end{equation}
To do so, observe that $E\left[( h_1(X,W,Y) + h_2(X,Y) + h_3(X) )^2\right] <\infty $ by Assumption \ref{as: vector w}. Hence, it follows from Lemma \ref{lem: simple square approximation} that \eqref{eq: cons variance est 1} holds if
\begin{equation}\label{eq: cons variance est 2}
R_n:= \frac{1}{n}\sum_{i=1}^n (H_{1i} + H_{2i} + H_{3i} - h_1(X_i,W_i,Y_i) - h_2(X_i,Y_i) - h_3(X_i))^2 = o_P(1).
\end{equation}
In turn, by the triangle inequality, $R_n \leq 9(R_{1n} + R_{2n} + R_{3n})$,
where
\begin{align*}
& R_{1n} := \frac{1}{n}\sum_{i=1}^n (H_{1i} - h_1(X_i,W_i,Y_i))^2,\\
& R_{2n} := \frac{1}{n}\sum_{i=1}^n (H_{2i} - h_2(X_i,Y_i))^2, \quad R_{3n}:= \frac{1}{n}\sum_{i=1}^n (H_{3i} - h_3(X_i))^2.
\end{align*}
We bound these three terms in turn.

Regarding $R_{1n}$, we have $E[|R_Y(Y) - \rho R_X(X) - W'\beta|^4]<\infty$ and $E[| R_X(X) - W'\gamma |^4]<\infty$ by Assumption \ref{as: vector w}. Also,
$$
\frac{1}{n}\sum_{i=1}^n \left| (R_i^Y - \widehat\rho R_i^X - W_i'\widehat\beta) - (R_Y(Y_i) - \rho R_X(X_i) - W_i'\beta) \right|^4 = o_P(1)
$$
by Lemma \ref{lem: glivenko cantelli}, Theorem \ref{thm: AN for rank rank reg}, the proof of Theorem~\ref{thm: AN for rank rank reg all coeffs},\footnote{Note that the relevants parts of the proof of Theorem~\ref{thm: AN for rank rank reg all coeffs} do not require Assumption~\ref{as: variable upsilon}.} and Assumption \ref{as: vector w}. In addition,
$$
\frac{1}{n}\sum_{i=1}^n \left| (R_i^X - W_i'\widehat\gamma) - (R_X(X_i) - W_i'\gamma) \right|^4 = o_P(1)
$$
by Lemmas \ref{lem: gamma asy normality} and \ref{lem: glivenko cantelli} and Assumption \ref{as: vector w}. Hence, $R_{1n} = o_P(1)$ by Lemma \ref{lem: simple square approximation 2}.

Regarding $R_{2n}$, we have by the triangle inequality that
\begin{align*} 
\sqrt{R_{2n}}& \leq \sup_{x,y\in\mathbb R}\left| \frac{1}{n}\sum_{j=1}^n (I(y,Y_j) - \widehat\rho I(x,X_j) - W_j'\widehat\beta)(R_j^X - W_j'\widehat \gamma) - h_2(x,y) \right| \\
& \leq \sup_{x,y\in\mathbb R}\left| \frac{1}{n}\sum_{j=1}^n (I(y,Y_j) - \widehat\rho I(x,X_j) - W_j'\widehat\beta)(R_j^X - W_j'\widehat \gamma) - \widehat h_2(x,y) \right| \\
& \quad + \sup_{x,y\in\mathbb R}\left|\widehat h_2(x,y) - h_2(x,y)\right|,
\end{align*}
where
$$
\widehat h_2(x,y) := \frac{1}{n}\sum_{j=1}^n (I(y,Y_j) - \rho I(x,X_j) - W_j'\beta)(R_X(X_j) - W_j' \gamma), \quad\text{for all }x,y\in\mathbb R. 
$$
Also, $\sup_{x,y\in\mathbb R}\left|\widehat h_2(x,y) - h_2(x,y)\right| = o_P(1)$ under Assumption \ref{as: vector w} by Lemma \ref{lem: glivenko cantelli 2} and the triangle inequality. In addition, for all $j=1,\dots,n$, denote
$B_j(x,y) := I(y,Y_j) - \rho I(x,X_j) - W_j'\beta$, $\widehat B_j(x,y) := I(y,Y_j) - \widehat\rho I(x,X_j) - W_j'\widehat\beta$, for all $x,y\in\mathbb R$, $C_j := R_X(X_j) - W_j' \gamma$, and $\widehat C_j := R_j^X - W_j'\widehat \gamma$. Then $\sup_{x,y\in\mathbb R} \frac{1}{n}\sum_{j=1}^n (\widehat B_j(x,y) - B_j(x,y))^4 = o_P(1)$ by Theorems \ref{thm: AN for rank rank reg} and \ref{thm: AN for rank rank reg all coeffs} and Assumption \ref{as: vector w}; $\frac{1}{n}\sum_{j=1}^n (\widehat C_j - C_j)^4 = o_P(1)$ by Lemmas \ref{lem: gamma asy normality} and \ref{lem: glivenko cantelli} and Assumption \ref{as: vector w}; $\sup_{x,y\in\mathbb R} \frac{1}{n}\sum_{j=1}^n B_j(x,y)^4 = O_P(1) $ and $\frac{1}{n}\sum_{j=1}^n C_j^4 = O_P(1)$ by Assumption \ref{as: vector w}. Hence, by Lemma \ref{lem: simple square approximation 2},
\begin{align*}
&\sup_{x,y\in\mathbb R}\left| \frac{1}{n}\sum_{j=1}^n (I(y,Y_j) - \widehat\rho I(x,X_j) - W_j'\widehat\beta)(R_j^X - W_j'\widehat \gamma) - \widehat h_2(x,y) \right| \\
&\quad = \sup_{x,y\in\mathbb R} \left| \frac{1}{n}\sum_{j=1}^n \widehat B_j(x,y) \widehat C_j - B_j(x,y)C_j\right| \\
&\quad \leq \sup_{x,y\in\mathbb R}\sqrt{\frac{1}{n}\sum_{j=1}^n (\widehat B_j(x,y) \widehat C_j - B_j(x,y)C_j)^2} =o_P(1).
\end{align*}
Combining the presented bounds gives $R_{2n} = o_P(1)$.

Regarding $R_{3n}$, we have by the triangle inequality that
\begin{align*}
& \sqrt{R_{3n}} \leq \sup_{x\in\mathbb R}\left| \frac{1}{n}\sum_{j=1}^n (R_j^Y - \widehat\rho R_j^X - W_j'\widehat\beta)(I(x,X_j) - W_j'\widehat\gamma) - h_3(x) \right| \\
& \quad \leq \sup_{x\in\mathbb R}\left| \frac{1}{n}\sum_{j=1}^n (R_j^Y - \widehat\rho R_j^X - W_j'\widehat\beta)(I(x,X_j) - W_j'\widehat\gamma) - \widehat h_3(x) \right| + \sup_{x\in\mathbb R}|\widehat h_3(x) - h_3(x)|,
\end{align*}
where
$$
\widehat h_3(x) := \frac{1}{n}\sum_{j=1}^n (R_Y(Y_j) - \rho R_X(X_j) - W_j'\beta)(I(x,X_j) - W_j'\gamma),\quad\text{for all }x\in\mathbb R.
$$
Also, $\sup_{x\in\mathbb R}|\widehat h_3(x) - h_3(x)| = o_P(1)$ under Assumption \ref{as: vector w} by Lemma \ref{lem: glivenko cantelli 2} and the triangle inequality. In addition, for all $j=1,\dots,n$, denote $B_j := R_Y(Y_j) - \rho R_X(X_j) - W_j'\beta, \ \widehat B_j := R_j^Y - \widehat \rho R_j^X - W_j'\widehat\beta$, $C_j(x) := I(x,X_j) - W_j' \gamma$, $\widehat C_j(x) := I(x,X_j) - W_j'\widehat \gamma$, for all $x\in\mathbb R$. Then $\frac{1}{n}\sum_{j=1}^n (\widehat B_j - B_j)^4 = o_P(1)$ by Theorems \ref{thm: AN for rank rank reg} and \ref{thm: AN for rank rank reg all coeffs}, Lemma \ref{lem: glivenko cantelli}, and Assumption \ref{as: vector w}; $\sup_{x\in\mathbb R}\frac{1}{n}\sum_{j=1}^n (\widehat C_j(x) - C_j(x))^4 = o_P(1)$ by Lemma \ref{lem: gamma asy normality} and Assumption \ref{as: vector w}; $\frac{1}{n}\sum_{j=1}^n B_j^4 = O_P(1)$ and $\sup_{x\in\mathbb R}\frac{1}{n}\sum_{j=1}^n C_j(x)^4 = O_P(1)$ by Assumption \ref{as: vector w}. Hence, by Lemma \ref{lem: simple square approximation 2},
\begin{align*}
&\sup_{x\in\mathbb R}\left| \frac{1}{n}\sum_{j=1}^n (R_j^Y - \widehat\rho R_j^X - W_j'\widehat\beta)(I(x,X_j) - W_j'\widehat\gamma) - \widehat h_3(x) \right| \\
&\qquad = \sup_{x\in\mathbb R} \left| \frac{1}{n}\sum_{j=1}^n \widehat B_j \widehat C_j(x) - B_jC_j(x)\right| \\
&\qquad\leq \sup_{x\in\mathbb R}\sqrt{\frac{1}{n}\sum_{i=1}^n (\widehat B_j \widehat C_j(x) - B_jC_j(x))^2} =o_P(1).
\end{align*}
Combining presented bounds gives $R_{3n} = o_P(1)$.

Thus, \eqref{eq: cons variance est 2} and hence \eqref{eq: cons variance est 1} are satisfied. In turn, combining \eqref{eq: cons variance est 1} with the law of large numbers yields
\begin{equation}\label{eq: cons variance est 3}
\frac{1}{n}\sum_{i=1}^n(H_{1i} + H_{2i} + H_{3i})^2 \to_P E[(h_1(X,W,Y) + h_2(X,Y) + h_3(X))^2].
\end{equation}
Also, $\hat\sigma_{\nu}^2 \to_P \sigma_{\nu}^2$ by Lemma \ref{lem: denominator limit}. The asserted claim now follows from combining these convergence results with the continuous mapping theorem.
\end{proof}

\subsection{Proofs for Section \ref{sec: other regressions}}

\subsubsection{Proofs for Section \ref{sec: rank rank reg with subpops}}

\begin{proof}[Proof of Theorem \ref{thm: AN for rank rank reg with subpops}]
    Fix $g \in \{1,\dots,n_G\}$ and let $\tilde{R}_{i,g}^X := \ind\{G_i=g\}R_i^X$, $\tilde{R}_{i,g}^Y := \ind\{G_i=g\}R_i^Y$, and $\tilde{W}_{i,g} := \ind\{G_i=g\}W_i$. Then, $\hat{\rho}_g$ is the OLS estimator from a regression of $\tilde{R}_{i,g}^Y$ on $\tilde{R}_{i,g}^X$ and $\tilde{W}_{i,g}$ using all observations $i=1,\ldots,n$. Therefore, as in the proof of Theorem~\ref{thm: AN for rank rank reg},
    \begin{align*}
        \sqrt n(\hat\rho_g - \rho_g) &= \frac{\frac{1}{\sqrt n}\sum_{i=1}^n \left(\tilde{R}_{i,g}^Y - \rho_g\tilde{R}^X_{i,g}\right)\left(\tilde{R}_{i,g}^X - \tilde{W}_{i,g}'\hat\gamma_g\right)}{\frac{1}{n}\sum_{i=1}^n \left(\tilde{R}_{i,g}^X - \tilde{W}_{i,g}'\hat\gamma_g\right)^2}\\
            &= \frac{\frac{1}{\sqrt n}\sum_{i=1}^n \ind\{G_i=g\}\left(\hat{R}_Y(Y_i) - \rho_g\hat{R}_X(X_i) \right)\left(\hat{R}_X(X_i) - W_i'\hat\gamma_g\right)}{\frac{1}{n}\sum_{i=1}^n \ind\{G_i=g\}\left(\hat{R}_X(X_i) - W_i'\hat\gamma_g\right)^2},
    \end{align*}
see \eqref{eq: fwl to be referred in another thm}.
    By Lemmas \ref{lem: denominator limit with subpops} and \ref{lem: three i's bound with subpops}, we have
    \begin{multline*}
        \sqrt n(\hat\rho_g - \rho_g) = \frac{\frac{1}{\sqrt n}\sum_{i=1}^n\ind\{G_i=g\}\left(\hat R_Y(Y_i) - \rho_g\hat R_X(X_i) - W_i'\beta_g\right)\left(\hat R_X(X_i) - W_i'\gamma_g\right)}{\frac{1}{n}\sum_{i=1}^n \ind\{G_i=g\}\left(\hat{R}_X(X_i) - W_i'\hat\gamma_g\right)^2},
    \end{multline*}
    up to an additive $o_P(1)$ term. Define $Z_i := (G_i,Y_i,X_i,W_i')'$ for all $i=1,\dots,n$ and
    $$f_g(Z_i,Z_j,Z_k):=\ind\{G_i=g\}\left(I(Y_j, Y_i) - \rho_g I(X_j, X_i) - W_i'\beta_g\right)(I(X_k, X_i) - W_i'\gamma_g)$$ 
    for all $i,j,k=1,\dots,n$. Also, define
    $h_g(Z_i,Z_j,Z_k) := 6^{-1}\sum_{i_1,i_2,i_3}f_g(Z_{i_1},Z_{i_2},Z_{i_3}),$
    where the sum is over all six permutations $(i_1,i_2,i_3)$ of the triplet $(i,j,k)$. Note that $h_g$ is a symmetric function satisfying $E[h_g(Z_i,Z_j,Z_k)]=0$ whenever $1\leq i<j<k\leq n$. Also, Assumption~\ref{as: vector w with subpops} implies that $E(h_g(Z_i,Z_j,Z_k)^2)<\infty$. Then, letting $z:=(\bar g,y,x,w')'$,
    \begin{align*} 
        &E[f_g(Z_1,Z_2,Z_3)\mid Z_1=z] = E[f_g(Z_1,Z_3,Z_2) | Z_1=z]\\
                &\quad= \ind\{\bar{g}=g\}(R_Y(y) -\rho_g R_X(x)-w'\beta_g)(R_X(x)-w'\gamma_g) = h_{1,g}(\bar{g},x,w,y),\\
        &E[f_g(Z_2,Z_1,Z_3)\mid Z_1=z] = E[f_g(Z_3,Z_1,Z_2)\mid Z_1=z]\\
                &\quad= E\left[\ind\{G=g\}(I(y,Y) -\rho_g I(x,X)-W'\beta_g)(R_X(X)-W'\gamma_g)\right] = h_{2,g}(x,y),\\
        &E[f_g(Z_2,Z_3,Z_1)\mid Z_i=z] = E[f_g(Z_3,Z_2,Z_1)\mid Z_1=z]\\
                &\quad= E\left[\ind\{G=g\}(R_Y(Y) -\rho_g R_X(X)-W'\beta_g)(I(x,X)-W'\gamma_g)\right] = h_{3,g}(x),
    \end{align*}
    we can argue as in the proof of Theorem~\ref{thm: AN for rank rank reg} that
    \begin{multline*}
        \frac{1}{n}\sum_{i=1}^n\ind\{G_i=g\}\left(\hat R_Y(Y_i) - \rho_g\hat R_X(X_i) - W_i'\beta_g\right)\left(\hat R_X(X_i) - W_i'\gamma_g\right)\\
            =\frac{1}{n}\sum_{i=1}^n \Big\{  h_{1,g}(X_i,W_i,Y_i) + h_{2,g}(X_i,Y_i) + h_{3,g}(X_i) \Big\} + o_P(n^{-1/2}).
    \end{multline*}
    The remainder of the proof then follows that of Theorem~\ref{thm: AN for rank rank reg}.
\end{proof}

\begin{lemma}\label{lem: denominator limit with subpops}
    Under Assumption \ref{as: vector w with subpops}, we have for any $g=1,\ldots,n_G$,
    $$n^{-1}\sum_{i=1}^n\ind\{G_i=g\}(\hat R_X(X_i)-W_i'\hat\gamma_g)^2\to_P E[\ind\{G_i=g\}(R_X(X)-W'\gamma_g)^2] = \sigma_{\nu,g}^2.$$
\end{lemma}

\begin{proof}
    Similar to that of Lemma \ref{lem: denominator limit}.
\end{proof}

\begin{lemma}\label{lem: three i's bound with subpops}
    Under Assumption \ref{as: vector w with subpops}, we have for any $g=1,\ldots,n_G$,
    \begin{align*}
    & \frac{1}{\sqrt n}\sum_{i=1}^n\ind\{G_i=g\}(\hat R_Y(Y_i) - \rho_g\hat R_X(X_i))(\hat R_X(X_i) - W_i'\hat\gamma_g) \\
    &\qquad = \frac{1}{\sqrt n}\sum_{i=1}^n\ind\{G_i=g\}(\hat R_Y(Y_i) - \rho_g\hat R_X(X_i) - W_i'\beta_g)(\hat R_X(X_i) - W_i'\gamma_g) + o_P(1).
    \end{align*}
\end{lemma}

\begin{proof}
    Similar to that of Lemma \ref{lem: three i's bound}.
\end{proof}

\subsubsection{Proofs for Section \ref{sec: outcome on rank}}

\begin{proof}[Proof of Theorem \ref{thm: AN for outcome on rank reg}] 
    As in the proof of Theorem \ref{thm: AN for rank rank reg}, 
    $$
    \sqrt n(\hat\rho - \rho) = \frac{\frac{1}{\sqrt n}\sum_{i=1}^n (Y_i - \rho\hat R_X(X_i))(\hat R_X(X_i) - W_i'\hat\gamma)}{\frac{1}{n}\sum_{i=1}^n (\hat R_X(X_i) - W_i'\hat\gamma)^2}.
    $$
    By Lemmas \ref{lem: denominator limit} and \ref{lem: three i's bound other reg}, we have
    \begin{equation}\label{eq: rho hat ratio other reg}
        \sqrt n(\hat\rho - \rho) = \frac{\frac{1}{\sqrt n}\sum_{i=1}^n (Y_i - \rho\hat R_X(X_i) - W_i'\beta)(\hat R_X(X_i) - W_i'\gamma)}{\frac{1}{n}\sum_{i=1}^n (\hat R_X(X_i) - W_i'\hat\gamma)^2} + o_P(1).
    \end{equation}
    Define $Z_i := (Y_i,X_i,W_i')'$ for all $i=1,\dots,n$ and
    $$f(Z_i,Z_j,Z_k):=\left( Y_i - \rho I(X_j, X_i) - W_i'\beta\right)(I(X_k, X_i) - W_i'\gamma)$$ 
    for all $i,j,k=1,\dots,n$.
    Also, define
    $h(Z_i,Z_j,Z_k) := 6^{-1}\sum_{i_1,i_2,i_3}f(Z_{i_1},Z_{i_2},Z_{i_3}),$
    where the sum is taken over all six permutations $(i_1,i_2,i_3)$ of the triplet $(i,j,k)$. Since $E[\varepsilon^4]<\infty$ and Assumption \ref{as: vector w} imply $E[h(Z_i,Z_j,Z_k)^2]<\infty$, we can argue as in the proof of Theorem \ref{thm: AN for rank rank reg} to show that
    \begin{multline}\label{eq: proj other reg}
        \frac{1}{n}\sum_{i=1}^n (Y_i - \rho\hat R_X(X_i) - W_i'\beta)(\hat R_X(X_i) - W_i'\gamma)\\
         = \frac{1}{n}\sum_{i=1}^n\Big\{ h_1(X_i,W_i,Y_i) + h_2(X_i) + h_3(X_i)\Big\} + o_P(n^{-1/2})
    \end{multline}
    where for $z=(x,y,w')'$, we define
    \begin{align*} 
        E[f(Z_i,Z_j,Z_k)\mid Z_i=z] &= E[f(Z_i,Z_k,Z_j) | Z_i=z]\\
                &= (y -\rho R_X(x)-w'\beta)(R_X(x)-w'\gamma) = h_1(x,w,y),\\
        E[f(Z_j,Z_i,Z_k)\mid Z_i=z] &= E[f(Z_k,Z_i,Z_j)\mid Z_i=z]\\
                &= E\left[(Y -\rho I(x,X)-W'\beta)(R_X(X)-W'\gamma)\right] = h_2(x),\\
        E[f(Z_j,Z_k,Z_i)\mid Z_i=z] &= E[f(Z_k,Z_j,Z_i)\mid Z_i=z]\\
                &= E\left[(Y -\rho R_X(X)-W'\beta)(I(x,X)-W'\gamma)\right] = h_3(x)
    \end{align*}
    Combining \eqref{eq: rho hat ratio other reg} with \eqref{eq: proj other reg} then implies the desired result by the same argument as that in the proof of Theorem~\ref{thm: AN for rank rank reg}.
\end{proof}

\begin{lemma}\label{lem: three i's bound other reg}
    Under Assumption \ref{as: vector w}, we have
    \begin{align*}
    & \frac{1}{\sqrt n}\sum_{i=1}^n( Y_i - \rho\hat R_X(X_i))(\hat R_X(X_i) - W_i'\hat\gamma) \\
    &\qquad = \frac{1}{\sqrt n}\sum_{i=1}^n(Y_i - \rho\hat R_X(X_i) - W_i'\beta)(\hat R_X(X_i) - W_i'\gamma) + o_P(1).
    \end{align*}
\end{lemma}

\begin{proof}
Similar to that of Lemma \ref{lem: three i's bound}.
\end{proof}

\subsubsection{Proofs for Section \ref{sec: rank on regressor}}

Denote by $W_l$ the $l$-th element of $W:=(W_1,\ldots,W_p)'$ and by $W_{-l}$ the vector of all elements of $W$ except the $l$-th. The projection of $W_l$ onto $W_{-l}$ for any $l=1,\ldots,p$ now takes the following form:
\begin{equation}\label{eq: first stage general regressor}
    W_l = W_{-l}'\gamma_l + \nu_l,\quad E[\nu_l W_{-l}] = 0,
\end{equation}
where $\gamma_1,\ldots,\gamma_p$ are $(p-1)$-dimensional vectors of parameters.

\begin{proof}[Proof of Theorem \ref{thm: AN for rank on regressor reg}]
Here, we prove that
    $$
    \sqrt n( \hat\beta - \beta) = \frac{1}{\sqrt n}\sum_{i=1}^n \psi_i + o_P(1)\to_D N(0,\Sigma),
    $$
    where $\Sigma  := E[\psi_i\psi_i']$,    $\psi_i := (\phi_1(W_i,Y_i), \ldots, \phi_p(W_i,Y_i)'$ for all $i=1,\dots,n$, 
    $\phi_l(w,y) := \sigma_{\nu_l}^{-2} \left[ h_{l,1}(w,y) + h_{l,2}(y)\right], $
    $\sigma_{\nu_l}^2 := E[\nu_l^2]$, and 
    \begin{align*}
        h_{l,1}(w,y) &:= (R_Y(y) - w'\beta)(w_l - w_{-l}'\gamma_l),\\
        h_{l,2}(y) & := E[(I(y,Y) - W'\beta)(W_l - W_{-l}'\gamma_l)],
    \end{align*}
    for all $w\in\mathbb R^p$ and $y\in\mathbb R$.
These formulas not only show the asymptotic normality of $\widehat\beta$ but also give an explicit expression for the asymptotic variance $\Sigma$.

    As in the proof of Theorem \ref{thm: AN for rank rank reg}, 
    $$
    \sqrt n(\hat\beta_l - \beta_l) = \frac{\frac{1}{\sqrt n}\sum_{i=1}^n (\hat{R}_Y(Y_i) - \beta_l W_{l,i})( W_{l,i} - W_{-l,i}'\hat\gamma_l)}{\frac{1}{n}\sum_{i=1}^n ( W_{l,i} - W_{-l,i}'\hat\gamma_l)^2}
    $$
    with
$
        \hat\gamma = (\sum_{i=1}^n W_{-l,i} W_{-l,i}')^{-1}(\sum_{i=1}^n W_{-l,i} W_{l,i}).
$
    By Lemmas \ref{lem: denominator limit other reg 2} and \ref{lem: three i's bound other reg 2}, we have
    \begin{equation}\label{eq: rho hat ratio other reg 2}
        \sqrt n(\hat\beta_l - \beta_l) = \frac{\frac{1}{\sqrt n}\sum_{i=1}^n (\hat{R}_Y(Y_i) - W_i'\beta)(W_{l,i} - W_{-l,i}'\gamma_l)}{\frac{1}{n}\sum_{i=1}^n (W_{l,i} - W_{-l,i}'\hat\gamma_l)^2} + o_P(1).
    \end{equation}
    Define $Z_i := (Y_i,W_i')'$ for all $i=1,\dots,n$,
    $$f(Z_i,Z_j):=\left( I(Y_j,Y_i) - W_i'\beta\right)(W_{l,i} - W_{-l,i}'\gamma_l)$$ 
    for all $i,j=1,\dots,n$, and $h(Z_i,Z_j) := 2^{-1}(f(Z_i,Z_j)+f(Z_j,Z_i))$ for all $i,j=1,\dots,n$. Since Assumption \ref{as: vector w} implies $E[h(Z_i,Z_j)^2]<\infty$, we can argue as in the proof of Theorem~\ref{thm: AN for rank rank reg} that
    \begin{equation}\label{eq: proj other reg 2}
        \frac{1}{n}\sum_{i=1}^n (\hat{R}_Y(Y_i) - W_i'\beta)(W_{l,i} - W_{-l,i}'\gamma_l) = \frac{1}{n}\sum_{i=1}^n\Big\{ h_{l,1}(W_i,Y_i) + h_{l,2}(Y_i) \Big\} + o_P(n^{-1/2})
    \end{equation}
    where for $z=(y,w')'$, we define
    \begin{align*} 
        E[f(Z_i,Z_j)\mid Z_i=z] &= (R_Y(y) -w'\beta)(w_l-w_{-l}'\gamma_l) = h_{l,1}(w,y),\\
        E[f(Z_j,Z_i)\mid Z_i=z] &= E\left[(I(y,Y) -W'\beta)(W_l-W_{-l}'\gamma_l)\right] = h_{l,2}(y).
    \end{align*}
    Combining \eqref{eq: rho hat ratio other reg 2} with \eqref{eq: proj other reg 2} then implies the desired result by the same argument as that in the proof of Theorem~\ref{thm: AN for rank rank reg}.
\end{proof}

\begin{lemma}\label{lem: denominator limit other reg 2}
    Under Assumption \ref{as: vector w}, for each $l=1,\ldots,p$, we have $\sigma_{\nu_l}^2>0$ and
    $$\frac{1}{n}\sum_{i=1}^n( W_{l,i}-W_{-l,i}'\hat\gamma_l)^2\to_P E[(W_l-W_{-l}'\gamma_l)^2] = \sigma_{\nu_l}^2.$$
\end{lemma}

\begin{proof}
    Since it follows from Assumption~\ref{as: vector w} that $\sigma_{\nu_l}^2>0$, the proof is similar to that of Lemma~\ref{lem: denominator limit}.
\end{proof}

\begin{lemma}\label{lem: three i's bound other reg 2}
    Under Assumption \ref{as: vector w}, for each $l=1,\ldots,p$, we have
    \begin{align*}
    & \frac{1}{\sqrt n}\sum_{i=1}^n( \hat{R}_Y(Y_i) - \beta_l W_{l,i})( W_{l,i} - W_{-l,i}'\hat\gamma_l) \\
    &\qquad  = \frac{1}{\sqrt n}\sum_{i=1}^n(\hat{R}_Y(Y_i) - W_i'\beta)( W_{l,i} - W_{-l,i}'\gamma_l) + o_P(1).
    \end{align*}
\end{lemma}

\begin{proof}
    Similar to that of Lemma \ref{lem: three i's bound}.
\end{proof}

\subsection{Proofs for Appendix \ref{sec: all coefficients}}
\begin{proof}[Proof of Theorem \ref{thm: AN for rank rank reg all coeffs}]
Fix $l=1,\dots,p$. As in the proof of Theorem \ref{thm: AN for rank rank reg},
$$
\sqrt n(\hat\beta_l - \beta_l) = \frac{\frac{1}{\sqrt n}\sum_{i=1}^n (\hat R_Y(Y_i) - \beta_l W_{l,i})(W_{l,i} - \hat\tau_l\hat R_X(X_i) - W_{-l,i}'\hat\delta_l)}{\frac{1}{n}\sum_{i=1}^n (W_{l,i} - \hat\tau_l\hat R_X(X_i) - W_{-l,i}'\hat\delta_l)^2},
$$
where
$$
\begin{pmatrix}
\hat{\tau_l}\\
\hat{\delta_l}
\end{pmatrix}=\left(\sum_{i=1}^{n}\begin{pmatrix}
\hat R_X(X_i)\\
W_{-l,i}
\end{pmatrix}\begin{pmatrix}
\hat R_X(X_i) & W_{-l,i}'\end{pmatrix}\right)^{-1}\sum_{i=1}^{n}\begin{pmatrix}
\hat R_X(X_i)\\
W_{-l,i}
\end{pmatrix}W_{l,i}.
$$
Further, again by the same arguments as those in the proof of Theorem \ref{thm: AN for rank rank reg},
\begin{align*}
& \frac{1}{\sqrt n}\sum_{i=1}^n (\hat R_Y(Y_i) - \beta_l W_{l,i})(W_{l,i} - \hat\tau_l\hat R_X(X_i) - W_{-l,i}'\hat\delta_l)\\
&\qquad = \frac{1}{\sqrt n}\sum_{i=1}^n (\hat R_Y(Y_i) - \rho\hat R_X(X_i) - W_{i}'\beta)(W_{l,i} - \tau_l\hat R_X(X_i) - W_{-l,i}'\delta_l) + o_P(1)\\
&\qquad = \frac{1}{\sqrt n}\sum_{i=1}^n \Big\{ h_{l,1}(X_i,W_i,Y_i) + h_{l,2}(X_i,Y_i) + h_{l,3}(X_i) \Big\} + o_P(1)
\end{align*}
and
$
n^{-1}\sum_{i=1}^n (W_{l,i} - \hat\tau_l\hat R_X(X_i) - W_{-l,i}'\hat\delta_l)^2 \to_P \sigma_{v_l}^2.
$
Combining these results with the central limit theorem and Slutsky's lemma gives the asserted claim.
\end{proof}

\begin{proof}[Proof of Theorem \ref{thm: AN for rank rank reg all coeffs subpops}]
Similar to that of Theorem \ref{thm: AN for rank rank reg with subpops}, based on the same relationship as that between Theorem \ref{thm: AN for rank rank reg all coeffs} and Theorem \ref{thm: AN for rank rank reg}.
\end{proof}

\begin{proof}[Proof of Theorem \ref{thm: AN for rank rank reg all coeffs gen outcome on rank}]
Similar to that of Theorem \ref{thm: AN for outcome on rank reg}, based on the same relationship as that between Theorem \ref{thm: AN for rank rank reg all coeffs} and Theorem \ref{thm: AN for rank rank reg}.
\end{proof}

\section{Bootstrap Validity}\label{app: bootstrap}
In this section, we discuss consistency of the nonparametric bootstrap
for estimating the asymptotic distribution of $\sqrt{n}(\widehat{\rho}-\rho)$
as mentioned in Remark \ref{rem: bootstrap}. To this end, let $\widehat\rho^*$ be the nonparametric bootstrap version of $\widehat\rho$ appearing in Section \ref{sec: asymptotic normality}, and let $\sigma^2$ be the same as in Theorem \ref{thm: AN for rank rank reg}. We have the following result.
\begin{theorem}\label{thm: bootstrap}
Suppose that \eqref{eq: model with controls}--\eqref{eq: first stage} hold and that Assumptions \ref{as: random sample}--\ref{as: variable nu} are satisfied. Then the sequence $\sqrt n(\widehat \rho^* - \widehat \rho)$ converges conditionally in distribution to $N(0,\sigma^2)$, given $\{(X_i,W_i,Y_i)\}_{i=1}^n$ in probability.
\end{theorem}

\begin{proof}
For simplicity, we focus on the case where $W=1$. The case with covariates is conceptually
similar but requires more notation. 

Fix $\omega\in[0,1]$. For all $t\in\mathbb R$, let $f_{t,1}\colon\mathbb R^2\to\mathbb R$ and $f_{t,2}\colon\mathbb R^2\to\mathbb R$ be the functions defined by 
$$
f_{t,1}(x,y):= \omega\mathds{1}\{x\leq t\} + (1-\omega)\mathds{1}\{x<t\},\quad (x,y)\in\mathbb R^2,
$$
$$
f_{t,2}(x,y):= \omega\mathds{1}\{y\leq t\} + (1-\omega)\mathds{1}\{y<t\},\quad (x,y)\in\mathbb R^2.
$$
Also, let $\mathcal K_1 := \{f_{t,1}\colon t\in\mathbb R\}$, $\mathcal K_2 := \{f_{t,2}\colon t\in\mathbb R\}$ and
$$
\mathcal K_3 := \{f\in\ell^{\infty}(\mathbb R^2)\colon f(x,y) = g_1(x)g_2(y)\text{ for all }(x,y)\in\mathbb R^2\text{ and some }(g_1,g_2)\in\mathcal M^2\},
$$
where $\mathcal M$ is the class of all non-decreasing functions mapping $\mathbb R$ to $[0,1]$. In addition, let $P$ denote the distribution of the pair $(X,Y)$. 

In the next lemma,  whose proof can be found at the end of this section, we show that the function class $\mathcal K := \mathcal K_1 \cup \mathcal K_2 \cup \mathcal K_3$ is $P$-Donsker; see Chapter 19 in \cite{vaart} for relevant definitions.
\begin{lemma}\label{lem: donsker}
The function class $\mathcal K$ is $P$-Donsker.
\end{lemma}
Next, let
$$
\mathcal F_1 := \{h_1 \in \ell^{\infty}(\mathcal K_1)\colon h_1(f_{t,1}) = g(t) \text{ for all }t\in\mathbb R\text{ and some }g\in\mathcal M\},
$$
$$
\mathcal F_2 := \{h_2 \in \ell^{\infty}(\mathcal K_2)\colon h_2(f_{t,2}) = g(t) \text{ for all }t\in\mathbb R\text{ and some }g\in\mathcal M\},
$$
$$
\mathcal F_3 := \left\{h_3\in \ell^{\infty}(\mathcal K_3)\colon h_3(f) = \int_{\mathbb R^2}f(x,y)dG(x,y)\text{ for all }f\in\mathcal K_3\text{ and some }G\in\mathcal G\right\},
$$
where $\mathcal G$ is the set of all cumulative distribution functions (cdfs) on $\mathbb R^2$ and the integral is understood in the Lebesgue sense. Also, let
$$
\mathcal D_1 := \{h_1\in\ell^{\infty}(\mathcal K_1)\colon h_1(f_{t,1})=g(t)\text{ for all }t\in\mathbb R\text{ and some }g\in\mathcal D\},
$$
$$
\mathcal D_2 := \{h_2\in\ell^{\infty}(\mathcal K_2)\colon h_2(f_{t,2})=g(t)\text{ for all }t\in\mathbb R\text{ and some }g\in\mathcal D\},
$$
where $\mathcal D$ is the set of cadlag functions mapping $\mathbb R$ to $\mathbb R$. In addition, let $\psi\colon \mathcal F_1\times \mathcal F_2 \to \mathcal K_3$ be the function defined by
$$
\psi(h_1,h_2)(x,y) := h_1(f_{x,1})h_2(f_{y,2}),\quad (x,y)\in\mathbb R^2,\quad (h_1,h_2)\in\mathcal F_1\times\mathcal F_2.
$$
Now, observe that since $\mathcal K = \mathcal K_1 \cup \mathcal K_2 \cup \mathcal K_3$, any function $h\in\ell^{\infty}(\mathcal K)$ can be decomposed into $(h_1,h_2,h_3)$, where $h_1\in\ell^{\infty}(\mathcal K_1)$, $h_2\in\ell^{\infty}(\mathcal K_2)$, and $h_3\in\ell^{\infty}(\mathcal K_3)$. Using this decomposition, define
$$
\mathbb D_{\phi} := \{h=(h_1,h_2,h_3)\in\ell^{\infty}(\mathcal K)\colon h_1\in\mathcal F_1,h_2\in\mathcal F_2,h_3\in\mathcal F_3\},
$$
$$
\mathbb D_0 := \{h = (h_1,h_2,h_3)\in\mathcal \ell^{\infty}(\mathcal K)\colon h_1\in\mathcal D_1,h_2\in\mathcal D_2,h_3\in \ell^{\infty}(\mathcal K_3)\},
$$
and let $\phi\colon\mathbb D_{\phi}\to\mathbb R$ be the function defined by
$$
\phi(h) := \frac{h_{3}(\psi(h_{1},h_{2})) - h_3(\psi(h_1,1))h_3(\psi(1,h_2))}{h_{3}(\psi(h_{1}^2,1)) - h_3(\psi(h_1,1))^2},\quad h=(h_{1},h_{2},h_{3})\in\mathbb{D}_{\phi}.
$$
Also, let $\theta := (\theta_1,\theta_2,\theta_3)\in\mathbb D_\phi$ be defined by $\theta_1(f_{t,1}) := E[f_{t,1}(X,Y)] = R_X(t)$ for all $t\in\mathbb R$, $\theta_2(f_{t,2}) := E[f_{t,2}(X,Y)] = R_Y(t)$ for all $t\in\mathbb R$, and 
$$
\theta_3(f) := \int_{\mathbb R^2} f(x,y)dF(x,y) = E[f(X,Y)]
$$ 
for all $f\in \mathcal K_3$, where $F\in\mathcal G$ is the cdf of the pair $(X,Y)$. Observe here that $\theta_3(\psi(\theta_1,\theta_2)) = E[R_X(X)R_Y(Y)]$, $\theta_3(\psi(\theta_1,1))=E[R_X(X)]$, $\theta_3(\psi(1,\theta_2))=E[R_Y(Y)]$, and $\theta_3(\psi(\theta_1^2,1))=E[R_X(X)^2]$, so that $\rho = \phi(\theta)$.

In the next lemma, whose proof can be found at the end of this section,  we show that the function $\phi$ is Hadamard differentiable at $\theta\in\mathbb D_{\phi}$ tangentially to $\mathbb D_0$; see Chapter 20 in \cite{vaart} for relevant definitions.
\begin{lemma}\label{lem: hadamard differentiability}
Suppose that Assumption \ref{as: variable nu} is satisfied. Then the function $\phi$ is Hadamard differentiable at $\theta\in\mathbb{D}_{\phi}\subset \ell^{\infty}(\mathcal K)$
tangentially to $\mathbb{D}_{0}\subset \ell^{\infty}(\mathcal K)$ with derivative $\phi_{\theta}'\colon\mathbb{D}_{0}\to\mathbb{R}$
given by
\begin{align*}
\phi_{\theta}'(h) & :=\frac{Cov(h_1(f_{X,1}),R_Y(Y)) + Cov(R_X(X),h_2(f_{Y,2})) - 2\rho Cov(h_1(f_{X,1}),R_X(X))}{Var(R_X(X))} \\
& \quad + \frac{h_3(\psi(\theta_1,\theta_2)) - \rho h_3(\psi(\theta_1^2,1)) - E[R_X(X)]h_3(\psi(1,\theta_2))}{Var(R_X(X))}\\
 &\quad + \frac{(2\rho E[R_X(X)] - E[R_Y(Y)])h_3(\psi(\theta_1,1))}{Var(R_X(X))},\  h=(h_{1},h_{2},h_{3})\in\mathbb{D}_0.
\end{align*}
\end{lemma}
\begin{remark}[Existence of the Derivative]
Observe that the functions $x\mapsto h_1(f_{x,1})$ and $y\mapsto h_2(f_{y,2})$ are bounded and, being cadlag, measurable. Hence, the terms $Cov(h_1(f_{X,1}),R_Y(Y))$, $Cov(R_X(X),h_2(f_{Y,2}))$, and $Cov(h_1(f_{X,1}),R_X(X))$ are well-defined. Also, $Var(R_X(X))>0$ under Assumption \ref{as: variable nu}. Thus, the derivative $\phi_\theta'$ is well-defined.
\QEDA
\end{remark}

Next, let $\hat\theta := (\hat\theta_1,\hat\theta_2,\hat\theta_3)\in\mathbb D_{\phi}$ be defined by 
$$
\hat\theta_1(f_{t,1}) := \int_{\mathbb R^2} f_{t,1}(x,y)\hat F(x,y) = \frac{1}{n}\sum_{i=1}^n f_{t,1}(X_i,Y_i) = \hat R_X(t) - \frac{1-\omega}{n},\quad t\in\mathbb R,
$$
$$
\hat\theta_2(f_{t,2}) := \int_{\mathbb R^2} f_{t,2}(x,y)\hat F(x,y) = \frac{1}{n}\sum_{i=1}^n f_{t,2}(X_i,Y_i) = \hat R_Y(t) - \frac{1-\omega}{n},\quad t\in\mathbb R,
$$
$$
\hat\theta_3(f) := \int_{\mathbb R^2} f(x,y)\hat F(x,y) = \frac{1}{n}\sum_{i=1}^n f(X_i,Y_i),\quad f\in\mathcal K_3,
$$
where $\hat F\in\mathcal G$ is the empirical cdf of the pair $(X,Y)$. Also, let $\hat\theta^* := (\hat\theta^*_1,\hat\theta^*_2,\hat\theta^*_3)\in\mathbb D_{\phi}$ be defined as $\hat\theta = (\hat\theta_1,\hat\theta_2,\hat\theta_3)$ with the empirical bootstrap cdf $\hat F^*\in\mathcal G$ replacing the empirical cdf $\hat F$.\footnote{To formally define the empirical bootstrap cdf, let $\{(X_i^*,Y_i^*)\}_{i=1}^n$ be a bootstrap sample obtained from drawing $n$ pairs from $\{(X_i,Y_i)\}_{i=1}^n$ at random with replacement. Then the empirical bootstrap cdf is defined by $\hat F^*(x,y) = n^{-1}\sum_{i=1}^n \mathds 1\{X_i^*\leq x,Y_i^*\leq y\}$, $(x,y)\in\mathbb R^2$.} 

Then $\hat\rho = \phi(\hat\theta)$ and $\hat\rho^* = \phi(\hat\theta^*)$. Also, it follows from Lemma \ref{lem: donsker} that $\sqrt n(\hat \theta - \theta)\to_D \mathbb G_P$ in $\ell^{\infty}(\mathcal K)$, where $\mathbb G_P$ is a tight Gaussian process on $\mathcal K$; see Chapter 19.2 in \cite{vaart}. In addition, $\mathbb G_P$ takes values in $\mathbb D_0$ by Theorem 19.3 in \cite{vaart}. Moreover, since all functions in $\mathcal K$ take values in $[0,1]$, the function class $\mathcal K$ has a bounded envelope. Thus, applying Theorem 23.9 in \cite{vaart} with $\mathbb D = \ell^{\infty}(\mathcal K)$ and $T = \mathbb G_P$ and using Theorem 23.7 in \cite{vaart} and Lemma \ref{lem: hadamard differentiability} above, it follows that the sequence $\sqrt n(\hat\rho^* - \hat\rho)$ converges in distribution to $\phi_\theta'(\mathbb G_P)$, given $\{(X_i,W_i,Y_i)\}_{i=1}^n$ in probability. On the other hand, it follows from Theorem 20.8 in \cite{vaart} that $\sqrt n(\hat\rho - \rho)\to_D \phi_\theta'(\mathbb G_P)$, and so the distribution of $ \phi_\theta'(\mathbb G_P)$ is equal to $N(0,\sigma^2)$ by Theorem \ref{thm: AN for rank rank reg}. Combining these results yields the asserted claim and completes the proof of the theorem.
\end{proof}

\begin{proof}[Proof of Lemma \ref{lem: donsker}]
Consider the following function classes:
$$
\mathcal K_{1,1} := \{f\in\ell^{\infty}(\mathbb R^2)\colon f(x,y) = \mathds 1\{x\leq t\}\text{ for all }(x,y)\in\mathbb R^2\text{ and some }t\in\mathbb R\},
$$
$$
\mathcal K_{1,2} := \{f\in\ell^{\infty}(\mathbb R^2)\colon f(x,y) = \mathds 1\{x < t\}\text{ for all }(x,y)\in\mathbb R^2\text{ and some }t\in\mathbb R\},
$$
$$
\mathcal K_{2,1} := \{f\in\ell^{\infty}(\mathbb R^2)\colon f(x,y) = \mathds 1\{y\leq t\}\text{ for all }(x,y)\in\mathbb R^2\text{ and some }t\in\mathbb R\},
$$
$$
\mathcal K_{2,2} := \{f\in\ell^{\infty}(\mathbb R^2)\colon f(x,y) = \mathds 1\{y < t\}\text{ for all }(x,y)\in\mathbb R^2\text{ and some }t\in\mathbb R\},
$$
$$
\mathcal K_{3,1} := \{f\in\ell^{\infty}(\mathbb R^2)\colon f(x,y) = g(x)\text{ for all }(x,y)\in\mathbb R^2\text{ and some }g\in\mathcal M\},
$$
$$
\mathcal K_{3,2} := \{f\in\ell^{\infty}(\mathbb R^2)\colon f(x,y) = g(y)\text{ for all }(x,y)\in\mathbb R^2\text{ and some }g\in\mathcal M\}.
$$
Observe that the function classes $\mathcal{K}_{1,1}$ and
$\mathcal{K}_{2,1}$ have finite bracketing integrals by Example 19.6
in Van der Vaart. Also, $\mathcal{K}_{1,2}$ and $\mathcal{K}_{2,2}$
have finite bracketing integrals by the same argument. In addition,
$\mathcal{K}_{3,1}$ and $\mathcal{K}_{3,2}$ have finite bracketing integrals
by Example 19.11 in Van der Vaart.

Now, to show that $\mathcal{K}_{1}$ has a finite bracketing integral,
fix $\varepsilon>0$ and consider $\varepsilon$-brackets $(l_{1,1},u_{1,1}),\dots,(l_{1,N(\varepsilon,1)},u_{1,N(\varepsilon,1)})$
 in $L_{2}(P)$ for the function class $\mathcal{K}_{1,1}$,
where we denoted $N(\varepsilon,1)=N_{[]}(\varepsilon,\mathcal{K}_{1,1},L_{2}(P))$.
Similarly, let $(l_{2,1},u_{2,1}),\dots,(l_{2,N(\varepsilon,2)},u_{2,N(\varepsilon,2)})$
be $\varepsilon$-brackets in $L_{2}(P)$ for the function class $\mathcal{K}_{1,2}$,
where $N(\varepsilon,2)=N_{[]}(\varepsilon,\mathcal{K}_{1,2},L_{2}(P))$.
Then for any $f\in\mathcal{K}_{1}$, there exist $g_{1}\in\mathcal{K}_{1,1}$
and $g_{2}\in\mathcal{K}_{1,2}$ such that $f=\omega g_{1}+(1-\omega)g_{2}$
and, correspondingly, $j(1)\in\{1,\dots,N(\varepsilon,1)\}$ and $j(2)\in\{1,\dots,N(\varepsilon,2)\}$
such that $l_{1,j(1)}(x,y)\leq g_{1}(x,y)\leq u_{1,j(1)}(x,y)$ and
$l_{2,j(2)}(x,y)\leq g_{2}(x,y)\leq u_{2,j(2)}(x,y)$ for all $x,y\in\mathbb{R}$.
On the other hand,
\begin{align*}
 & \|\omega(u_{1,j(1)}-l_{1,j(1)})+(1-\omega)(u_{2,j(2)}-l_{2,j(2)})\|_{P,2}\\
 & \quad\leq\omega\|(u_{1,j(1)}-l_{1,j(1)})\|_{P,2}+(1-\omega)\|(u_{2,j(2)}-l_{2,j(2)})\|_{P,2}\leq\varepsilon.
\end{align*}
Therefore, pairs of functions of the form $(\omega l_{1,j(1)}+(1-\omega)l_{2,j(2)},\omega u_{1,j(1)}+(1-\omega)u_{2,j(2)})$
with $j(1)\in\{1,\dots,N(\varepsilon,1)\}$ and $j(2)\in\{1,\dots,N(\varepsilon,2)\}$
form $\varepsilon$-brackets for $\mathcal{K}_{1}$ in $L_{2}(P)$.
Thus, 
\[
N_{[]}(\varepsilon,\mathcal{K}_{1},L_{2}(P))\leq N_{[]}(\varepsilon,\mathcal{K}_{1,1},L_{2}(P))\times N_{[]}(\varepsilon,\mathcal{K}_{1,2},L_{2}(P)),
\]
and so $\mathcal K_{1}$ has a finite bracketing integral:
\begin{align*}
 & \int_{0}^{1}\sqrt{\log N_{[]}(\varepsilon,\mathcal{K}_{1},L_{2}(P))}d\varepsilon\\
 & \quad\leq\int_{0}^{1}\sqrt{\log N_{[]}(\varepsilon,\mathcal{K}_{1,1},L_{2}(P))}d\varepsilon+\int_{0}^{1}\sqrt{\log N_{[]}(\varepsilon,\mathcal{K}_{1,2},L_{2}(P))}d\varepsilon<\infty.
\end{align*}
Moreover, $\mathcal{K}_{2}$ has a finite bracketing integral by the
same argument.

Further, to show that $\mathcal{K}_{3}$ has a finite bracketing integral
as well, fix $\varepsilon>0$ and consider $\varepsilon$-brackets $(l_{1,1},u_{1,1}),\dots,(l_{1,N(\varepsilon,1)},u_{1,N(\varepsilon,1)})$  in $L_{2}(P)$ for the function class $\mathcal{K}_{3,1}$,
where $N(\varepsilon,1)=N_{[]}(\varepsilon,\mathcal{K}_{3,1},L_{2}(P))$.
Similarly, let $(l_{2,1},u_{2,1}),\dots,(l_{2,N(\varepsilon,2)},u_{2,N(\varepsilon,2)})$
be $\varepsilon$-brackets in $L_{2}(P)$ for the function class $\mathcal{K}_{3,2}$,
where $N(\varepsilon,2)=N_{[]}(\varepsilon,\mathcal{K}_{3,2},L_{2}(P))$.
Then for any $f\in\mathcal{K}_{3}$, there exist $g_{1}\in\mathcal{K}_{3,1}$
and $g_{2}\in\mathcal{K}_{3,2}$ such that $f=g_{1}g_{2}$ and, correspondingly,
$j(1)\in\{1,\dots,N(\varepsilon,1)\}$ and $j(2)\in\{1,\dots,N(\varepsilon,2)\}$
such that $l_{1,j(1)}(x,y)\leq g_{1}(x,y)\leq u_{1,j(1)}(x,y)$ and
$l_{2,j(2)}(x,y)\leq g_{2}(x,y)\leq u_{2,j(2)}(x,y)$ for all $x,y\in\mathbb{R}$.
On the other hand,
\begin{align*}
 & \|u_{1,j(1)}u_{2,j(2)}-l_{1,j(1)}l_{2,j(2)}\|_{P,2}\\
 & \quad\leq\|(u_{1,j(1)}-l_{1,j(1)})u_{2,j(2)}\|_{P,2}+\|l_{1,j(1)}(u_{2,j(2)}-l_{2,j(2)})\|_{P,2}\\
 & \quad\leq\|u_{1,j(1)}-l_{1,j(1)}\|_{P,2}+\|u_{2,j(2)}-l_{2,j(2)}\|_{P,2}\leq2\varepsilon.
\end{align*}
Therefore, pairs of functions of the form $(l_{1,j(1)}l_{2,j(2)},u_{1,j(1)}u_{2,j(2)})$
with $j(1)\in\{1,\dots,N(\varepsilon,1)\}$ and $j(2)\in\{1,\dots,N(\varepsilon,2)\}$
form $(2\varepsilon)$-brackets for $\mathcal{K}_{3}$ in $L_{2}(P)$.
Thus, 
\[
N_{[]}(\varepsilon,\mathcal{K}_{3},L_{2}(P))\leq N_{[]}(\varepsilon/2,\mathcal{K}_{3,1},L_{2}(P))\times N_{[]}(\varepsilon/2,\mathcal{K}_{3,2},L_{2}(P)),
\]
and so $\mathcal K_3$ has a finite bracketing integral:
\begin{align*}
 & \int_{0}^{1}\sqrt{\log N_{[]}(\varepsilon,\mathcal{K}_{3},L_{2}(P))}d\varepsilon\\
 & \quad\leq2\int_{0}^{1}\sqrt{\log N_{[]}(\varepsilon,\mathcal{K}_{3,1},L_{2}(P))}d\varepsilon+2\int_{0}^{1}\sqrt{\log N_{[]}(\varepsilon,\mathcal{K}_{3,2},L_{2}(P))}d\varepsilon<\infty.
\end{align*}
Finally, since the bracketing integral of the union of function classes
does not exceed the sum of bracketing integrals for individual function
classes, it follows that $\mathcal{K}=\mathcal{K}_{1}\cup\mathcal{K}_{2}\cup\mathcal{K}_{3}$
has a finite bracketing integral as well. The asserted claim now follows
from Theorem 19.5 in \cite{vaart}.
\end{proof}

\begin{proof}[Proof of Lemma \ref{lem: hadamard differentiability}]
Let $\phi_{1}\colon\mathbb{D}_{\phi}\to\mathbb{R}$, $\phi_{2}\colon\mathbb{D}_{\phi}\to\mathbb{R}$, $\phi_{3}\colon\mathbb{D}_{\phi}\to\mathbb{R}$, and $\phi_{4}\colon\mathbb{D}_{\phi}\to\mathbb{R}$
be the functions defined by
$$
\phi_{1}(h):=h_{3}(\psi(h_{1},h_{2})),\quad \phi_{2}(h):=h_{3}(\psi(h_{1}^2,1)),\quad \phi_{3}(h):=h_{3}(\psi(h_{1},1)),\quad \phi_{4}(h):=h_{3}(\psi(1,h_{2})),
$$
so that 
$$
\phi(h)=\frac{\phi_{1}(h) - \phi_3(h)\phi_4(h)}{\phi_{2}(h) - \phi_3(h)^2},\quad h = (h_1,h_2,h_3)\in\mathbb D_{\phi}.
$$
We proceed in three steps.

{\bf Step 1.} Here, we show that $\phi_{1}$ is Hadamard differentiable
at $\theta\in\mathbb{D}_{\phi}$ tangentially to $\mathbb{D}_0$ with
derivative $\phi_{\theta,1}'\colon\mathbb{D}_{0}\to\mathbb{R}$ given
by
\[
\phi_{\theta,1}'(h):=E[h_1(f_{X,1})R_Y(Y) + R_X(X)h_2(f_{Y,2})]+h_{3}(\psi(\theta_1,\theta_2)),\ h=(h_{1},h_{2},h_{3})\in\mathbb{D}_{0}.
\]
To do so, note
that $\phi_{\theta,1}'$ is linear by standard properties of Lebesgue
integrals (the integrals exist because the functions $x\mapsto h_1(f_{x,1})$ and $y\mapsto h_2(f_{y,2})$ are bounded and, being cadlag, measurable) and also bounded: for any $h=(h_{1},h_{2},h_{3})\in\mathbb{D}_{0}$,
\[
|\phi_{\theta,1}'(h)|\leq\|h_{1}\|_{\infty}+\|h_{2}\|_{\infty}+\|h_{3}\|_{\infty}\leq 3\|h\|_{\infty}.
\]
Next, consider any sequences $\{t_{n}\}_{n\geq1}\subset\mathbb{R}$ and $\{h_{n}\}_{n\geq1}\subset\ell^{\infty}(\mathcal K)$ such that $t_{n}\to0$
in $\mathbb{R}$ and $h_{n}\to h\in\mathbb{D}_{0}$ in $\ell^{\infty}(\mathcal K)$
as $n\to\infty$ and $\theta+t_{n}h_{n}\in\mathbb{D}_{\phi}$ for
all $n\geq1$. For convenience, denote $h_{n}=(h_{n,1},h_{n,2},h_{n,3})$
for all $n\geq1$ and $h=(h_{1},h_{2},h_{3})$. Decompose
\[
\phi_{1}(\theta+t_{n}h_{n})-\phi_{1}(\theta)=I_{n,1}+I_{n,2},
\]
where
\[
I_{n,1}:=(\theta_3+t_{n}h_{n,3})(\psi(\theta_1+t_{n}h_{n,1},\theta_2+t_{n}h_{n,2}))-(\theta_3+t_{n}h_{n,3})(\psi(\theta_1,\theta_2)),
\]
\[
I_{n,2}:=(\theta_3+t_{n}h_{n,3})(\psi(\theta_1,\theta_2))-\theta_3(\psi(\theta_1,\theta_2)).
\]
Further, observe that since $\theta+t_{n}h_{n}\in\mathbb{D}_{\phi}$
implies that $\theta_3+t_{n}h_{n,3}\in\mathcal{F}_3$, it follows that for
each $n\geq1$, there exists $F_{n}\in\mathcal{G}$ such that
\[
(\theta_3+t_{n}h_{n,3})(\psi(h_{1},h_{2}))=\int_{\mathbb{R}^2}h_{1}(f_{x,1})h_{2}(f_{x,2})dF_{n}(x, y),\quad (h_{1},h_{2})\in\mathcal{F}_1\times\mathcal{F}_2.
\]
Therefore, $I_{n,1}=I_{n,1,1}+I_{n,1,2}+I_{n,1,3}$, where
\begin{align*}
I_{n,1,1} & :=t_{n}\int_{\mathbb{R}^2}h_{n,1}(f_{x,1})R_Y(y)dF_{n}(x,y),\\
I_{n,1,2} & :=t_{n}\int_{\mathbb{R}^2}R_X(x)h_{n,2}(f_{y,2})dF_{n}(x,y),\\
I_{n,1,3} & :=t_{n}^{2}\int_{\mathbb{R}^2}h_{n,1}(f_{x,1})h_{n,2}(f_{y,2})dF_{n}(x,y).
\end{align*}
Here,
\[
|I_{n,1,3}|\leq t_{n}^{2}\|h_{n,1}\|_{\infty}\|h_{n,2}\|_{\infty}=o(t_{n})
\]
since $\|h_{n,1}-h_{1}\|_{\infty}\to0$, $\|h_{n,2}-h_{2}\|_{\infty}\to0$,
$\|h_{1}\|_{\infty}<\infty$, and $\|h_{2}\|_{\infty}<\infty$. Also,
\begin{align*}
 & \left|I_{n,1,1}-t_{n}\int_{\mathbb{R}^2}h_{n,1}(f_{x,1})R_{Y}(y)dF(x,y)\right|\\
 & \quad=t_{n}\left|(\theta_3+t_{n}h_{n,3})(\psi(h_{n,1},\theta_2))-\theta_3(\psi(h_{n,1},\theta_2))\right|=t_{n}^{2}|h_{n,3}(\psi(h_{n,1},\theta_2))|=o(t_{n})
\end{align*}
since $\|h_{n,3}-h_{3}\|_{\infty}\to0$ and $\|h_{3}\|_{\infty}<\infty$.
In addition,
$$
t_n\left| \int_{\mathbb{R}^2}h_{n,1}(f_{x,1})R_{Y}(y)dF(x,y) - \int_{\mathbb{R}^2}h_{1}(f_{x,1})R_{Y}(y)dF(x,y) \right| \leq t_n \|h_{n,1} - h_1\|_{\infty} = o(t_n)
$$
since $\|h_{n,1} - h_1\|_{\infty} \to 0$, so that
$$
\left|I_{n,1,1}-t_{n}\int_{\mathbb{R}^2}h_{1}(f_{x,1})R_{Y}(y)dF(x,y)\right| = o(t_n).
$$
Similarly,
\[
\left|I_{n,1,2}-t_{n}\int_{\mathbb{R}^2}R_{X}(x)h_{2}(f_{y,2})dF(x,y)\right|=o(t_{n}).
\]
Moreover,
\[
|I_{n,2}-t_{n}h_{3}(\psi(\theta_1,\theta_2))|=t_{n}|h_{n,3}(\psi(\theta_1,\theta_2))-h_{3}(\psi(\theta_1,\theta_2))|=o(t_{n}),
\]
since $\|h_{n,3}-h_{3}\|_{\infty}\to0$. Thus, it follows that
\[
\left|\frac{\phi_{1}(\theta+t_{n}h_{n})-\phi_{1}(\theta)}{t_{n}}-\phi_{\theta,1}'(h)\right|\to0,\quad\text{as }n\to\infty,
\]
which implies that $\phi_{1}$ is Hadamard differentiable at $\theta\in\mathbb{D}_{\phi}$
tangentially to $\mathbb{D}_{0}$ with derivative $\phi_{\theta,1}'$.

{\bf Step 2.} Observe that it follows from the same argument as that in
Step 1 that $\phi_{2}$ is Hadamard differentiable at $\theta\in\mathbb{D}_{\phi}$
tangentially to $\mathbb{D}_{0}$ with derivative $\phi_{\theta,2}'\colon\mathbb{D}_{0}\to\mathbb{R}$
given by
\[
\phi_{\theta,2}'(h):=2E[h_1(f_{X,1})R_X(X)]+h_{3}(\psi(\theta_1^2,1)),\quad h=(h_{1},h_{2},h_{3})\in\mathbb{D}_0.
\]
Similarly, $\phi_3$ is Hadamard differentiable at $\theta\in\mathbb{D}_{\phi}$
tangentially to $\mathbb{D}_{0}$ with derivative $\phi_{\theta,3}'\colon\mathbb{D}_{0}\to\mathbb{R}$
given by
\[
\phi_{\theta,3}'(h):=E[h_1(f_{X,1})]+h_{3}(\psi(\theta_1,1)),\quad h=(h_{1},h_{2},h_{3})\in\mathbb{D}_0
\]
and $\phi_4$ is Hadamard differentiable at $\theta\in\mathbb{D}_{\phi}$
tangentially to $\mathbb{D}_{0}$ with derivative $\phi_{\theta,4}'\colon\mathbb{D}_{0}\to\mathbb{R}$
given by
\[
\phi_{\theta,4}'(h):=E[h_2(f_{Y,2})]+h_{3}(\psi(1,\theta_2)),\quad h=(h_{1},h_{2},h_{3})\in\mathbb{D}_0.
\]

{\bf Step 3.} Here, we complete the proof. Denote $\mathcal C := \{(v,x,y,z)\in\mathbb R^4\colon x>y^2\}$
and let $c\colon\mathcal C\to\mathbb{R}$ be the function defined by $c(v,x,y,z) := (v - yz)/(x - y^2)$ for all $(v,x,y,z)\in\mathcal C$, with $c_1'(v,x,y,z)$, $c_2'(v,x,y,z)$, $c_3'(v,x,y,z)$, and $c_4'(v,x,y,z)$ denoting the corresponding partial derivatives. Then observe that Assumption \ref{as: variable nu} implies that $\phi_{2}(\theta) - \phi_3(\theta)^2 = E[R_X(X)^2] - (E[R_X(X)])^2>0$. Thus, for
the same sequences $\{t_{n}\}_{n\geq1}\subset\mathbb{R}$ and $\{h_{n}\}_{n\geq1}\subset\ell^{\infty}(\mathcal K)$
as those in Step 1, $\phi_{2}(\theta+t_{n}h_{n}) - \phi_{3}(\theta+t_{n}h_{n})^2>0$ for sufficiently
large $n$ by Step 2, and so 
\begin{align*}
\phi(\theta+t_{n}h_{n}) & =c(\phi_{1}(\theta+t_{n}h_{n}),\dots,\phi_{4}(\theta+t_{n}h_{n}))\\
 & =c(\phi_{1}(\theta)+t_{n}\phi_{\theta,1}'(h)+o(t_{n}),\dots,\phi_{4}(\theta)+t_{n}\phi_{\theta,4}'(h)+o(t_{n}))\\
 & = \phi(\theta) + c'_{1}(\phi_{1}(\theta),\dots,\phi_{4}(\theta))t_{n}\phi_{\theta,1}'(h)+\dots+c_{4}'(\phi_{1}(\theta),\dots,\phi_{4}(\theta))t_{n}\phi_{\theta,4}'(h)+o(t_{n}),
\end{align*}
where the second line follows from Steps 1 and 2 and the third from
the Taylor theorem. Substituting here the derivatives of the function $c$, we obtain
\[
\left|\frac{\phi(\theta+t_{n}h_{n})-\phi(\theta)}{t_{n}}-\phi_{\theta}'(h)\right|\to0,\quad\text{as }n\to\infty.
\]
In addition, $\phi_{\theta}'$ is linear by the properties of Lebesgue
integrals and also bounded by the same arguments as those used above. Thus, $\phi$ is Hadamard differentiable at $\theta\in\mathbb{D}_{\phi}$ tangentially to $\mathbb{D}_{0}$ with derivative $\phi_{\theta}'$. The asserted claim follows.
\end{proof}

\end{appendix}

\end{document}